\documentclass[10pt]{article}
\usepackage{fullpage}
\usepackage{amsmath,amsfonts,amsthm,amssymb}
\usepackage{verbatim} 
\usepackage{hyperref}

\numberwithin{equation}{section}

\providecommand{\w}{W} 
\providecommand{\wt}{\w_{3D}} 
\providecommand{\ft}{f_{3D}} 
\renewcommand{\wr}{W_r} 
\providecommand{\ud}{\, \mathrm{d}}
\providecommand{\dxy}{\ud x'}
\providecommand{\dxyz}{\ud x}
\providecommand{\dz}{\ud x_3}
\providecommand{\dS}{\ud S}
\providecommand{\dr}{\ud r}
\providecommand{\grad}{D}
\providecommand{\eps}{{\varepsilon_*}}
\providecommand{\epse}{\varepsilon_*}
\providecommand{\R}{\mathbb{R}}
\providecommand{\dist}{\,\mathrm{dist}}
\providecommand{\Omegah}{{\Omega\times (-h/2,h/2)}} 
\providecommand{\A}{\mathcal{A}} 
\providecommand{\ut}{\tilde u}

\renewcommand{\lll}{\lambda_1,\lambda_2,\lambda_3}
\renewcommand{\ll}{\lambda_1,\lambda_2}
\providecommand{\bo}{w_{3D}}

\providecommand{\relaxed}{{\Omega_R}}
\providecommand{\nonrelaxed}{{\Omega_N}}
\providecommand{\NR}{N_\rho} 
\providecommand{\RR}{R_\rho} 
\providecommand{\B}{B_{h}^{3D}} 
\providecommand{\E}{E_h^{3D}} 
\renewcommand{\o}{w} 
\providecommand{\ug}{u_h^{3D}} 
\providecommand{\eg}{E_{h}^{3D}} 

\newcommand{\footnotedva}[1]{}

\DeclareMathOperator*{\argmin}{argmin\,}

\newtheorem{theorem}{Theorem}
\newtheorem{lemma}{Lemma}[section]
\newtheorem*{lemma*}{Lemma}
\newtheorem*{prop*}{Proposition}
\newtheorem*{cor*}{Corollary}

\newtheorem*{remark*}{Remark}
\newtheorem{remark}[lemma]{Remark}

\begin{document}

\title{Wrinkles as a relaxation of compressive stresses in an~annular thin film}
\author{Peter Bella\footnote{P. Bella gratefully acknowledges partial support from NSF through grant DMS-0807347.} ~and
Robert V. Kohn\footnote{R. V. Kohn gratefully acknowledges partial support from NSF grants DMS-0807347 and OISE-0967140.}\\
Courant Institute of Mathematical Sciences\\
New York University\\
bella@cims.nyu.edu, kohn@cims.nyu.edu}

\maketitle


\begin{abstract}
It is well known that an elastic sheet loaded in tension will wrinkle and that the length scale of the wrinkles tends to zero with vanishing thickness of the sheet  [Cerda and Mahadevan, Phys. Rev. Lett. 90, 074302 (2003)]. We give the first mathematically rigorous analysis of such a problem. Since our methods require an explicit understanding of the underlying (convex) relaxed problem, we focus on the wrinkling of an annular sheet loaded in the radial direction [Davidovitch et al., PNAS {\bf 108} (2011), no. 45]. 
%
Our main achievement is identification of the scaling law of the minimum energy as the thickness of the sheet tends to zero. This requires proving an upper bound and a lower bound that
scale the same way. We prove both bounds first in a simplified Kirchhoff-Love setting
and then in the nonlinear three-dimensional setting. To obtain the optimal upper bound, we need to adjust a naive construction (one family of wrinkles superimposed on a planar deformation) by introducing a cascade of wrinkles.
The lower bound is more subtle, since it must be ansatz-free. 
\end{abstract}


\section{Introduction}


In the last few years the wrinkling and folding of thin elastic sheets has attracted a lot
of attention in both the mathematics and physics communities (see, e.g., the recent book by Audoly and Pomeau~\cite{bib-audolypomeau}).
Wrinkled configurations can be viewed as (local) minimizers of a suitable elastic energy, 
consisting of a non-convex ``membrane energy'' plus a higher-order singular perturbation
representing ``bending energy''.
Though the physically relevant wrinkled configurations are local minimizers, 
we can begin to understand their character by focusing on (i) the minimum 
value of the elastic energy, and (ii) the properties of low-energy deformations. 
In this paper we identify the {\it scaling law} of the minimum energy for
an annular sheet stretched in the radial direction. 
This requires proving an upper bound and a lower bound that
scale the same way. A naive approach to the upper bound, based on a single length scale of wrinkling, fails to achieve the optimal scaling~\cite{bib-benny1}; the successful approach uses a cascade of wrinkles. The lower bound is more subtle, since it must be ansatz-free. We prove it  
first in a reduced Kirchhoff-Love setting and later in a general ansatz-free three-dimensional setting.

As mentioned above, the behavior of thin elastic sheets has attracted considerable attention from the physics community (see, e.g., work on sheets of graphene~\cite{bib-graphene}). Mahadevan and Cerda considered the stretching of a rectangular elastic sheet with clamped boundaries~\cite{bib-cerdamaha1} (see also~\cite{bib-mahariz} for experiments and~\cite{bib-friedl} for numerical computation), by minimizing the elastic energy of the membrane within a particular ansatz. 
The problem we consider here is similar, but our viewpoint and achievement are different: we prove an upper bound and a matching ansatz-free lower bound on the elastic energy. Our analysis does not assume a specific form for the solution.

Our treatment requires knowledge of the underlying convex relaxed problem. In the radial setting (see~\cite{bib-benny1}) the relaxed problem reduces to a simple one dimensional variational problem 
which can be analyzed quite completely. For this reason we focus on 
an annular sheet stretched in the radial direction as a convenient model problem, rather than addressing the case considered in~\cite{bib-cerdamaha1}.


One might ask why we are so interested in the scaling law of the minimal energy.
As mentioned above, stable configurations are local minimizers of the elastic energy. 
It seems difficult to find such configurations analytically
(the associated fourth-order PDE is highly nonlinear). But we expect the
physically-relevant configurations to have relatively low energy. Therefore we
can obtain some information about them by identifying the energy scaling law, then
investigating the properties of configurations that achieve this law. 
While the present paper focuses mainly on the energy scaling law, certain consequences
are immediately evident. In particular, since our scaling law is linear in $h$ and the
bending energy is $h^2 \int |\nabla ^2 u_3|^2$, it is immediately evident that the low-energy
configurations become increasingly complex as $h \rightarrow 0$.

\subsection{Context}

Motivated by experiments, several physics papers have studied the wrinkling of thin elastic films from a theoretical point of view. As already mentioned, Davidovitch et al.~\cite{bib-benny1} considered an annular film stretched in the radial direction (see also~\cite{bib-geminard} for related results). Dead loads applied both on the inside and outside boundary cause the film to wrinkle in some region. Indeed, if the loads inside are large enough compared to the loads on the outer boundary, the deformation in the radial direction forces the concentric circles of the material near the inner boundary to decrease their length by more than is required by the Poisson ratio of the material. Therefore, the membrane needs to waste this excess in the circumference either by compression or by buckling out of plane, contributing to the energy with some amount which depends on $h$.
In~\cite{bib-benny1} they found an optimal solution (using energy minimization methods) within a particular ansatz and using a linear stress-strain law, obtaining conclusions about the extent of wrinkled region and the period and amplitude of wrinkles. In the present paper we consider the same problem using a nonlinear $3D$ model. We will prove an upper bound and a matching lower bound without assuming any ansatz.

Our problem seems related to the experiment reported in~\cite{bib-huang}. It consists of a circular thin elastic film placed on a liquid substrate with a droplet on top of the film. In this case, the capillary forces at the boundary stretch the film in the radial direction and the capillary forces from the drop force the film to wrinkle. This experiment was studied theoretically in~\cite{bib-cerdavella} (see also~\cite{benny2011}), though a lot of questions still remain open. We believe that our methods may also be useful in the study of this problem.

The idea of proving an upper bound and a matching lower bound for the minimum of the energy has a long history; see e.g., the work of Kohn and M\"uller on a model for martensitic phase transformation~\cite{bib-bob+muller1} (see also~\cite{bib-contibranching} for subsequent progress). As in our setting, the energy in~\cite{bib-bob+muller1} was composed of a nonconvex function of $\grad u$ singularly perturbed by a higher order term. In the setting of~\cite{bib-bob+muller1} 
the minimizer develops a fine branching structure. Similar phenomena are seen 
in uniaxial ferromagnets and type I superconductors (see, e.g.,~\cite{bib-choksi-conti-bob-otto},~\cite{bib-choksikohnotto}, and~\cite{bib-lifschitz}). 
In all these settings there is a ``relaxed problem'', whose minimizers are the weak limits of optimal configurations as $h \to 0$. The minimal energy for $h > 0$ is that of the relaxed problem plus a small correction that scales with $h$.  
One difference here is the special character of the singular perturbation -- bending energy rather than surface energy -- which leads to creation of smooth wrinkles rather than walls. 

The main focus here is the scaling of wrinkles associated with $h>0$. This is different from mere identification of the extent of the wrinkled region, which can be done by studying the relaxed problem or using the tension field theory.
There is a lot of literature on this application of tension field theory (see, e.g., the 1961 NASA report~\cite{bib-nasa}, \cite{epstein-tension}, or a recent work on balloons~\cite{bib-baginski2011}).


\subsection{The main idea}

Before starting with rigorous arguments let us outline our result and the main ideas of its proof. We will work variationally, considering the 
sum of the elastic energy of the thin sheet and boundary terms representing the work done by the loads. We first consider a simplified two-dimensional setting where the elastic energy is further split into a membrane and a bending term defined in terms of the midplane deformation $u: \R^2 \to \R^3$. The membrane term is written as the integral of a reduced $2D$ stored energy density obtained in a systematic way from the original $3D$ stored energy density. As the bending term we choose the $L^2$ norm of the second derivatives of the out-of-plane displacement as often seen in the linear F\"oppl-von K\'arm\'an energy used for small slopes and deformations. This $2D$ model is a curious hybrid since we use a nonlinear stretching term together with a linear bending term. One can ask why we don't also use a linear stretching part? In fact, that scenario would be very limited since it would lead to a very restrictive linear model for the relaxed problem. 

Since our focus is the limiting behavior as the thickness $h$ of the sheet tends to zero, 
we divide the energy by $h$ to get the energy per unit thickness $E_h$. The first step toward identification of the scaling law is to separate the contributions to $E_h$ from wrinkling and from the bulk deformation. This is done by considering a relaxed problem, where instead of the original stored energy density we use its quasiconvexification and formally set $h = 0$ (see~\cite{bib-pipkin2} for more detail on this topic).

Under mild assumptions the relaxed energy is convex, smaller than the original energy, and independent of the thickness. Moreover, we show that it possesses a unique solution $u_0$ (up to a translation) which is radially symmetric and planar. 
We denote the relaxed energy of $u_0$ by $\mathcal{E}_0$; as we will see the energy $E_h(u_0)$ is strictly larger than $\mathcal{E}_0$. This is 
because $u_0$ involves compression in the hoop direction in a region close to the inner boundary (we will call it the ``relaxed'' region). Since the thickness is small the sheet prefers to wrinkle rather than to compress. 

The idea of the construction for the upper bound on the minimum energy is to superimpose wrinkles upon $u_0$. After optimizing the amplitude and wavelength of the wrinkles in a naive ansatz we obtain a solution with energy $\mathcal{E}_0 + Ch\left|\log h\right|$. To remove the logarithmic factor (i.e. to get the same scaling as the lower bound) we need to work harder. We observe that the out-of-plane part of the deformation decreases suddenly at the boundary of the relaxed region; this is the source of the $|\log h|$ factor. At the same time, since the amplitude of the out-of-plane deformation is vanishing, the bending is vanishing as well. Therefore we can introduce branching of the wrinkles (changing their period) near the boundary of the relaxed region; this increases the bending but at the same time decreases the amplitude of the out-of-plane deformation, making the decrease of the out-of-plane displacement less steep. By arguing this
way we obtain a construction whose energy is bounded by $\mathcal{E}_0 + Ch$.

The lower bound $\min E_h \ge \mathcal{E}_0 + ch$ ($c>0$ independent of $h$) is proved using an argument by contradiction. In the simplified two-dimensional setting, we first use the relaxed problem to prove an estimate on the out-of-plane displacement $u_3$. Then using interpolation we show smallness of $\grad u_3$, allowing us to project the solution into the plane without changing its energy too much. Finally, we compare this projection with $u_0$ and obtain a contradiction from an argument about the area of the deformed annulus.
The generalization to the nonlinear three-dimensional setting uses the main arguments of the 2D setting coupled with rigidity estimates first derived by Friesecke, James, and M\"uller~\cite{bib-rigid1}.

The paper is organized as follows. In Section~\ref{sec-model} we describe the three dimensional energy together with a reduced two dimensional model. The definition of the relaxed problem and a theorem about its unique minimizer is in Section~\ref{sec-relax}. Section~\ref{sec-2d} contains both the upper bound and the lower bound in the reduced 2D setting. In Section~\ref{sec3d} we generalize the upper and lower bounds proved in Section~\ref{sec-2d} to the three-dimensional setting. The last section contains a brief discussion of our achievements together with some open questions.

We will use notation $x = (x_1,x_2,x_3) = (x',x_3)$ for points in $\R^3$, $A:B = \mathrm{tr}(A^TB)$ for the Frobenius inner product on matrices, and $\partial_1$, $\partial_2$ for partial derivatives with respect to the first and second variable.

\section{The model}\label{sec-model}

We are interested in deformations of isotropic elastic thin films of annular shape. 
We consider a nonlinear three-dimensional elastic energy (per unit thickness) in a cylindrical domain with small thickness~$h$ 
\begin{equation*}
\E(u) = \frac{1}{h} \int_{\Omegah} \wt(Du) \dxyz.
\end{equation*}
The stored energy density $\wt(M)$ is assumed to be isotropic, so it can be written as a symmetric function of the eigenvalues of $\sqrt{M^TM}$:
\begin{equation*}
 \wt(M) = \ft(\lambda_1,\lambda_2,\lambda_3). 
\end{equation*}
Here and below, we assume that $M=\grad u$ has strictly positive determinant; this is natural, since $M$ is the gradient of an elastic deformation.


As already mentioned, we are interested in deformations of annular thin films. We consider a thin cylindrical domain $\Omegah$ with a cross-section
\begin{equation*}
 \Omega = \left\{ x \in \mathbb{R}^2 : R_{in} < |x| < R_{out} \right\}
\end{equation*}
for some radii $0 < R_{in} < R_{out}$. The dead loads are applied on the inner boundary in the radial direction (with magnitude $T_{in}$, pointing inwards) and on the outer boundary (with magnitude $T_{out}$, pointing outwards), so the film will mostly stretch in the radial direction. These loads contribute to the total energy as
\begin{equation*}
 \B(u) := \frac{T_{in}}{h} \int_{|\hat x|=R_{in}} u(x) \cdot\frac{\hat x}{R_{in}} \dS - \frac{T_{out}}{h} \int_{|\hat x|=R_{out}} u(x) \cdot\frac{\hat x}{R_{out}} \dS,
\end{equation*}
where $\hat x = (x_1,x_2,0)$ and $\dS$ denotes surface measure.
We will show in Theorem~\ref{thm-3d} that under suitable assumptions on the elastic energy density~$\wt$, radii $R_{in},R_{out}$, and forces $T_{in},T_{out}$, we have
\begin{equation}\label{eq12}
\min_{u} \E(u) + \B(u) = \mathcal{E}_0 + O(h),
\end{equation}
where $\mathcal{E}_0$ is a constant (depending on $\Omega, T_{in}, T_{out}$ and $\wt$). Since $\mathcal{E}_0$ is the limiting energy as $h \to 0$, we view it as representing
the ``bulk energy'' of the deformation. The order-$h$ correction is the contribution from the wrinkling of the membrane. 

\subsection{The reduced model}

Since we consider domains which are thin in the $x_3$-direction, we can gain insight 
by first considering a~reduced two dimensional Kirchhoff-Love model. In this setting we are interested only in the deformation of one cross section (e.g. the mid-plane $x_3 = 0$), knowing that we can extend the deformation to the thin three-dimensional body by assuming that straight lines normal to the plane remain straight and normal to the plane after deformation. Assuming this ansatz, the energy per unit thickness is the sum of the ``membrane'' and ``bending'' energies
\begin{equation*}
\int_\Omega \w(\grad u)\dxy + h^2 \int_\Omega \mathcal{Q}(\grad \nu) \dxy,
\end{equation*}
where $\nu$ is the normal to the mid-surface of the deformation, $\mathcal{Q}$ is a certain quadratic function (derived from $\wt$) and the form of $\w(\grad u)$ will be discussed in Section~\ref{sect-w}. The second term (the bending energy) can be expressed using the first and second derivative of $u$. In the 2D analysis we will replace the bending term $\mathcal{Q}(\grad \nu)$ by a simpler term $|D^2 u_3|^2$.
Though the new term doesn't represent a physically correct bending energy, it is mathematically more convenient and still captures the main phenomenon. After including the boundary terms the two-dimensional energy has the form:
\begin{equation}\label{eq001}
  E_h(u) := \int_\Omega \w(\grad u) + h^2 |\grad^2 u_3|^2 \dxy + B(u),
\end{equation}
where the boundary terms are 
\begin{equation}\label{eq003}
 B(u) :=  T_{in} \int_{|x'|=R_{in}}  u(x') \cdot\frac{x'}{R_{in}} \dS - T_{out} \int_{|x'|=R_{out}} u(x') \cdot\frac{x'}{R_{out}} \dS.
\end{equation}

As in the general $3D$ setting our main result is a scaling law for the minimum of the energy. We will show that if $E_h$ is defined by~\eqref{eq001} then
\begin{equation*}
 \min_u E_h(u) = \mathcal{E}_0 + O(h)
\end{equation*}
for any sufficiently small $h>0$. The constant $\mathcal{E}_0$ is the minimum of the relaxed energy (the same constant as in~\eqref{eq12}). The order-$h$ correction is the contribution from the wrinkling of the membrane.

\subsection{The energy density}\label{sect-w}

In this section we will describe the assumptions we impose on the elastic energy density $\wt$. 

Since the energy density~$\wt(M)$ is isotropic, it is convenient to represent it as a function of the principal strains (i.e., the eigenvalues $\lambda_1 \ge \lambda_2 \ge \lambda_3 \ge 0$ of $(M^TM)^{1/2}$):
\begin{equation}\nonumber
 \wt(M) = \ft(\lll).
\end{equation}
We assume that
\begin{equation}\label{eq031}
\begin{gathered}
 \ft \in C^2( [0,\infty)^3),\quad \ft(1,1,1) = 0,\quad \ft \ge 0, \quad 0 < D^2 \ft \le C, \\
 \ft(\lll) \ge C_0(\lambda_1^2 + \lambda_2^2 + \lambda_3^2)^{p/2} - C_1,
\end{gathered}
\end{equation}
where $p \in (1,2]$. 

The main motivation to consider a $p$-th power lower bound for the energy density for large strains rather than the quadratic growth is to include a broader range of materials. For example, Agostiniani et al.~\cite{bib-ago-dalmaso-desimone} showed that the growth condition in~\eqref{eq031} with $p=3/2$ is satisfied by both the {\it neo-Hookean} compressible model and {\it Mooney-Rivlin} compressible model, whereas these models do not satisfy quadratic growth for large matrices~(i.e.~\eqref{eq031} with $p=2$).

To define an elastic energy in the reduced two-dimensional setting (i.e. for a map $u : \R^2 \to \R^3$), we need to define a stored energy density as a function of $\grad u$, i.e. for $3\times 2$ matrices. 
One way to do this is to optimize the missing third component. 
We set
\begin{equation}\label{eq002}
 \w(M) := \min_{\xi \in \R^3} \wt(M|\xi),
\end{equation}
where $M \in \R^{3\times 2}$ and $M|\xi$ denotes a $3 \times 3$ matrix with first two columns identical with $M$ and $\xi$ as the third column. It turns out that if we write
\begin{gather}\nonumber
 \wt(F) = g(I_1,I_2,J),\\
 J := \det (F), \ \ C := F^TF,\ \  I_1 := J^{-2/3}\, \mathrm{tr}(C),\ \  I_2 := \frac{J^{-4/3}}{2}\left( (\mathrm{tr}(C))^2 - \mathrm{tr}(C^2) \right),\nonumber
\end{gather}
and if  
\begin{gather}\label{eq049}
 \frac{\partial g}{\partial I_1}(I_1,I_2,J) \ge 0, \quad \frac{\partial g}{\partial I_2}(I_1,I_2,J) \ge 0, \quad \frac{\partial g}{\partial I_1}(I_1,I_2,J) + \frac{\partial g}{\partial I_2}(I_1,I_2,J) > 0,
\end{gather}
then the $\xi$ that achieves the minimum in~\eqref{eq002} has to satisfy $\xi \perp M$.
For completeness, we give a proof of this fact in the Appendix (see Lemma~\ref{lm5}).


Isotropy of the energy density implies that $\w(M)$ is a symmetric function of eigenvalues of $\sqrt{M^TM}$:
\begin{align}\label{w=f}
 \w(M) &= f(\lambda_1,\lambda_2).
\end{align}
It is easy to see that the function $f$ is related to $\ft$. Let 
\begin{equation}\nonumber
 \bo(\ll) := \argmin_{t > 0} \ft(\ll,t).
\end{equation}
Then we immediately obtain
\begin{equation}\nonumber
 \w(M) = f(\ll) = \ft(\ll,\bo(\ll)).
\end{equation}
Moreover, the function $f$ inherits properties of $\ft$. Indeed, \eqref{eq031} implies that
\begin{equation}\label{eq032}
\begin{gathered}
 f\in C^2( [0,\infty)^2), \quad f(1,1)=0, \quad f \ge 0, \quad 0 < D^2 f \le C, \\
 f(\ll) \ge C_0(\lambda_1^2 + \lambda_2^2)^{p/2} - C_1.
\end{gathered} 
\end{equation}
In contrast with three dimensions, the case $p=2$ in~\eqref{eq032} is not very restrictive in two dimensions. For example, two-dimensional energy densities obtained from incompressible three-dimensional models have often quadratic growth at infinity, and so they satisfy quadratic lower bound for large strains. We will prove our main results assuming $p=2$ in two dimensions and $1 < p \le 2$ in three dimensions.

For a given $\lambda_1 > 1$ we define $w(\lambda_1)$ as the point of minimum for the function $f(\lambda_1,\cdot)$, i.e.
\begin{equation}\label{eqw}
 f(\lambda_1,w(\lambda_1)) = \min_{t > 0} f(\lambda_1,t);
\end{equation}
we call $w(\lambda_1)$ the {\it natural width} of the strip with first principal strain $\lambda_1$. We assume that for $\lambda_1 > 1$ 
\begin{equation}
\o(\lambda_1) \textrm{ is a differentiable and non-increasing function} \label{eq027}.
\end{equation}

We also assume that for $\lambda_1 > 1$ and $\lambda_2 > \o(\lambda_1)$ the following conditions hold:
\begin{gather}
 \left(\partial_1f(\ll) - \partial_2f(\ll)\right)\left(\lambda_1 - \lambda_2\right) \ge 0, \label{eq034}\\
 \partial_{11} f(\lambda_1,\o(\lambda_1)) + \partial_{12} f(\lambda_1,\o(\lambda_1))\o'(\lambda_1) > 0, \label{eq035} \\
 \partial_{12}f(\ll) \ge 0. \label{eq038}
\end{gather}
The meaning of these relations will become apparent in a moment. Briefly stated we use them to show convexity of the relaxed energy (see Section~\ref{sec-relax}). The strict inequality in~\eqref{eq035} is not a typo --- it is associated with strict convexity of the (relaxed) energy density in $2D$ in the tensile direction. 

Finally, we assume for $\lambda_1 > 1$ that
\begin{equation}
 \label{eq18}
\det D^2f(\lambda_1,\o(\lambda_1)) \cdot \left(\lambda_1 - \o(\lambda_1)\right) >  \ \partial_{12}f(\lambda_1,\o(\lambda_1)) \cdot \partial_1f(\lambda_1,\o(\lambda_1)). 
\end{equation}
Unlike~(\ref{eq034}-\ref{eq038}), this inequality does not seem to have a simple interpretation; however it is satisfied by typical choices of $f$ (e.g. the one associated with an incompressible neoHookean $3D$ model). Condition~\eqref{eq18} will be used in our analysis of the relaxed problem (Lemma \ref{lm2.2} and equation~\eqref{eqgrowth}).

\section{The relaxed problem}\label{sec-relax}

In this section we study relaxed problem and the properties of its minimizer. 

We define the relaxed energy density $\wr(M)$ as the quasiconvexification of $\w(M)$ (see Pipkin~\cite{bib-pipkin2} for more details) and the relaxed functional as
\begin{equation}\label{eq4}
 E_0(u) := \int_\Omega \wr(\grad u) \dxy + B(u),
\end{equation}
where $B(u)$ was defined in~\eqref{eq003}.

It will be crucial to our analysis that $W_r(M)$ is a convex function of the $3 \times 2$ matrix $M$. Pipkin proved in~\cite{bib-pipkin1} that this is true whenever the unrelaxed density $W(M)$ is a convex function of $M^TM$. This is true for a broad range of material models. We shall assume throughout this paper that
\[\wr(M) \textrm{ is a convex function of }M.\]
As with $\wt$ and $\w$, it is convenient to represent the relaxed density $\wr$ 
as a function of the principal strains: 
\begin{equation}\nonumber
 \wr(M) = f_r(\lambda_1,\lambda_2).
\end{equation}

We would like to write down $f_r$ explicitly. To do that we follow the idea of Pipkin~\cite{bib-pipkin2}. Using the natural width $\o(\lambda)$ defined by~\eqref{eqw}, we define 
\begin{gather}\label{eq033}
 f_m(\ll) = 
 \begin{cases}
 f(\ll) & \lambda_1 \ge \o(\lambda_2) \textrm{~ and ~} \lambda_2 \ge \o(\lambda_1), \\
 f(\lambda_1,\o(\lambda_1)) & \lambda_2 < \o(\lambda_1) \textrm{~ and ~} \lambda_1 > 1, \\
 f(\lambda_2,\o(\lambda_2)) & \lambda_1 < \o(\lambda_2) \textrm{~ and ~} \lambda_2 > 1, \\
 0 & \lambda_1 \le 1 \textrm{~ and ~} \lambda_2 \le 1,
 \end{cases} 
\end{gather}
and $W_m(M) := f_m(\lambda_1,\lambda_2)$. Pipkin showed in \cite{bib-pipkin2} that
$$W_r(M) \le W_m(M).$$
We will show in a moment that under our hypotheses $W_m(M)$ is convex, in particular $W_m(M) \le W_r(M)$, from which it follows immediately that $W_m(M) = W_r(M)$. 

We want to show that $W_m$ is convex. Pipkin~\cite{bib-pipkin2} showed that this is equivalent to showing that the function $f_m$ is convex and monotone in both variables:
\begin{equation}\label{c1} 
 \grad^2 f_m \ge 0, \qquad \partial_\alpha f_m \ge 0, \quad \alpha = 1,2
\end{equation}
and satisfies the ordered force condition
\begin{equation}\label{c2} 
 \left(\partial_1f_m(\ll) - \partial_2f_m(\ll)\right)\left(\lambda_1 - \lambda_2\right) \ge 0.
\end{equation}
(When $\wt$ is the energy density of the incompressible neo-Hookean material, $f_m$ takes a particularly simple form, studied in~\cite{bib-pipkin2}). 

It is natural to express~\eqref{c1} and~\eqref{c2} using some conditions on $f$. First, in the region where $\lambda_1 \ge \o(\lambda_2), \lambda_2 \ge \o(\lambda_1)$ condition~\eqref{c1} follows from~\eqref{eq032}, whereas the latter condition~\eqref{c2} is equivalent to~\eqref{eq034}.

In the second case of~\eqref{eq033}, we see from~\eqref{eq032} that
\begin{equation}\nonumber
 \partial_1 f_m(\ll) = \partial_1 f(\lambda_1,\o(\lambda_1)) + \underbrace{\partial_2 f(\lambda_1,\o(\lambda_1))}_0 \o'(\lambda_1)) = \partial_1 f(\lambda_1,\o(\lambda_1)) \ge 0,
\end{equation}
and obviously $\partial_2 f_m(\ll) = 0$; therefore~\eqref{c2} is satisfied in this case. By~\eqref{eq035} we see that 
\begin{equation}\label{eq035a}
 \partial_{11} f_m(\ll) = \partial_{11} f(\lambda_1,\o(\lambda_1)) + \partial_{12} f(\lambda_1,\o(\lambda_1))\o'(\lambda) > 0,
\end{equation}
and by definition of $f_m$ also $\partial_{12}f_r(\ll) = \partial_{22}f_r(\ll) = 0$. 
Since $f_m$ is symmetric, the third case follows immediately. 

Finally, in the last case $\lambda_1 \le 1, \lambda_2 \le 1,$ both~\eqref{c1} and~\eqref{c2} are trivially satisfied. We have shown that if $f$ satisfies~\eqref{eq032}, \eqref{eq034}, and~\eqref{eq035}, then $f_m$ satisfies~\eqref{eq033}, in particular $W_m$ is convex. Therefore $f_r = f_m$ and $f_r$ has the form~\eqref{eq033}.

\subsection{The one-dimensional variational problem}

We want to find a minimizer of the relaxed energy~\eqref{eq4}. Assuming it is radially symmetric, we formulate a one-dimensional variational problem, which admits a unique minimizer~$v$. Afterward we show some properties of $v$.


To look for a radially symmetric minimizer, we consider 
\begin{equation}\label{eq3}
 u_0(r,\theta) = (v(r),\theta)
\end{equation}
in polar coordinates. Then~\eqref{eq4} becomes the one-dimensional variational problem
\begin{equation}\label{eq13}
 \min_{v\in W^{1,p}(R_{in},R_{out})} \int_{R_{in}}^{R_{out}} r\cdot f_r(v'(r), v(r)/r) \ud r + R_{in}T_{in}v(R_{in}) - R_{out}T_{out}v(R_{out}).
\end{equation}

The function $f_r(\lambda_1,\lambda_2)$ is defined for $\lambda_1 > 0$ and $\lambda_2 > 0$. Since we do not assume a priori that $v' \ge 0$ or $v > 0$ a.e., we need to extend the domain of $f_r$. It is convenient to do it in the following way: 
\begin{equation}\label{eq041}
 f_r(\lambda_1,\lambda_2) := \begin{cases} f_r(\lambda_1,\o(\lambda_1)) & \textrm{ if } \lambda_1 > 0, \lambda_2 \le 0,\\
 f_r(|\lambda_1|,\lambda_2) & \textrm{ if } \lambda_1 < 0. \end{cases}
\end{equation}

Under our assumptions both $\wr$ and $f_r$ are convex functions, and so we expect this variational problem to be solvable using direct methods of Calculus of Variations provided $T_{in}R_{in} < T_{out}R_{out}$. 

\begin{remark}\label{rmk-nonrelaxed}
If $T_{in}R_{in} < T_{out}R_{out}$, we claim that any minimizer $v$ of~\eqref{eq13} has tensile hoop stress somewhere in $(R_{in},R_{out})$. Indeed, if we denote by $\sigma_r$ and $\sigma_\theta$ radial and hoop stress, respectively, the optimality condition reads $(r\sigma_r)' = \sigma_\theta$, $\sigma_r(R_{in})=T_{in}, \sigma_r(R_{out})=T_{out}$. Integrating the equation gives
\[ \int_{R_{in}}^{R_{out}} \sigma_\theta(r) \dr = T_{out}R_{out} - T_{in}R_{in} > 0,\]
and so $\sigma_\theta(r) > 0$ for some $r \in (R_{in},R_{out})$.
\end{remark}

\begin{remark}\label{rmk-rel}
In the case $T_{in}R_{in}=T_{out}R_{out}$, the Euler-Lagrange equation implies that the hoop stress is identically zero, and a minimizer of~\eqref{eq13} is unique only up to an additive constant (i.e. only $v'$ is uniquely determined). 

If $T_{in}R_{in} > T_{out}R_{out}$, it is easy to see that the energy in~\eqref{eq13} is not bounded from below and so the minimization problem has no solution.
\end{remark}

\begin{remark}\label{rmk-f}
It is important to understand what are the consequences of extension~\eqref{eq041}. We observe that $f_r$ being even in $\lambda_1$ means that $\partial_1 f_r(\lambda_1,\lambda_2) \le 0$ for $\lambda_1 \le 0$. Therefore, 
\[\partial_1 f_r(\lambda_1,\lambda_2) > 0\]
implies
\[ \lambda_1 > 0.\]
Also, for any $\lambda_2 < 0$, the hoop stress
\[ \sigma_\theta = \partial_2 f_r(\lambda_1,\lambda_2) = 0\]
no matter how large $|\lambda_2|$ is.  
\end{remark}

From now on we will always assume that 
\begin{equation}\label{eq030}
 T_{in}R_{in} < T_{out}R_{out}.
\end{equation}
We claim that under this condition there exists a unique solution $v$ to the variational problem~\eqref{eq13}.

\begin{theorem}\label{thm-rel}
Let $f$ satisfy {\rm (\ref{eq032}--\ref{eq18})} and  
%
let $0 < R_{in} < R_{out}$ and $T_{out}$ be fixed. Then there exists a range of inner-boundary loads $T_{in}$, a subset of $(T_{out},T_{out}R_{out}/R_{in})$, 
such that minimizer $v$ of~\eqref{eq13} exists, it is unique, and the following holds:
\begin{itemize}
\item There exists $L \in (R_{in},R_{out})$ such that 
\begin{equation*}
 \frac{v(r)}{r} < \o(v'(r))\quad r \in (R_{in},L),\qquad \frac{v(r)}{r} \ge \o( v'(r))\quad r \in  (L,R_{out}),
\end{equation*}
i.e. there is tensile hoop stress in $(L,R_{out})$ whereas in $(R_{in},L)$ there is compression in the hoop direction. We will call the region with a tensile hoop stress a {\rm non-relaxed} region and its complement a {\rm relaxed} region. 
\item The deformation $v$ avoids interpenetration, i.e.
\begin{equation}\nonumber 
 v(r) > 0, v'(r) >  0 \textrm{ for } r \in (R_{in},R_{out}).
\end{equation}
\item Consider the function $h(r) = \o( v'(r)) - \frac{v(r)}{r}$, representing the amount of arclength we need to waste in the relaxed region $(R_{in},L)$. Then 
\begin{equation}\label{eq-th2}
 h'(L) < 0
\end{equation}
(and obviously $h(L)=0$.)
\end{itemize}
\end{theorem}

\begin{remark*}
 Condition~\eqref{eq-th2} means that the excess arclength associated with wrinkling at radius $r$ grows linearly as a function of the distance from $L$; we will see later that this introduces some difficulties in the upper bound.
\end{remark*}
\begin{remark*}
 It is easy to show that the relaxed energy associated with an incompressible neo-Hookean material $W(\grad u) = f(\lambda_1,\lambda_2) = C(\lambda_1^2 + \lambda_2^2 + \lambda_1^{-2}\lambda_2^{-2} - 3)$ satisfies all assumptions of Theorem~\ref{thm-rel}. In the case of a material with a linear stress-strain law, the conclusions of Theorem~\ref{thm-rel} are much more easily seen in the geometrically linear setting by explicitly writing down the solution $v$ (see~\cite{bib-benny1}).
\end{remark*}

The proof of Theorem~\ref{thm-rel} consists of several steps. First, we prove the existence of a solution $v$ for~\eqref{eq13} (Lemma~\ref{rel_existence}). Next, we show some elementary properties of $v$ (Lemma~\ref{lm-pos}), which will allow us to show the uniqueness of $v$ (Lemma~\ref{rel_uniqueness}). 
Afterward we prove the remaining properties of $v$. This consists of the following steps:
\begin{itemize}
 \item we show that any relaxed interval has to start at $R_{in}$ (Lemma~\ref{lm2.2});
 \item we show $v(R_{in};T) > 0$ and $v(R_{in};T) < \o(v'(R_{in};T))R_{in}$ for some loads $T$ (Lemma~\ref{lm2.1});
 \item we prove that $(R_{in},R_{out})$ splits into a relaxed and a non-relaxed interval, both of them non-empty;
 \item we show~\eqref{eq-th2}.
\end{itemize}

\begin{lemma}\label{rel_existence}
 Under the assumptions of Theorem~\ref{thm-rel}, there exists a minimizer $v$ of~\eqref{eq13}. 
\end{lemma}

\begin{proof}
We start by rewriting~\eqref{eq13} in the form
\begin{equation*}
 \min_{v\in W^{1,p}(R_{in},R_{out})} \int_{R_{in}}^{R_{out}} r\cdot f_r( v'(r), v(r)/r) - (\varphi(r) v(r))' \ud r,
\end{equation*}
where $\varphi(r)\! =\! \frac{R_{out}-r}{R_{out}-R_{in}} R_{in}T_{in} + \frac{r-R_{in}}{R_{out}-R_{in}} R_{out}T_{out}$ is the linear interpolation between $R_{in}T_{in}$ and $R_{out}T_{out}$.
We observe that $(\varphi v)' = r\left(\varphi' \frac{v}{r} + \frac{\varphi}{r} v'\right)$.
Since $f=f_r$ for large strains by~\eqref{eq033} and $f$ has p-th power growth by~\eqref{eq032}, $f_r$ has also p-th power growth for large strains. Then the previous integral is bounded from below by 
\begin{equation*}
C + \int_{R_{in}}^{R_{out}} \left( |v'(r)|^p - \frac{\varphi(r)}{r} v'(r) + \left(\frac{v(r)}{r}\right)_+^p - \varphi'(r)\frac{v(r)}{r}\right) r \dr.
\end{equation*}
The assumption~\eqref{eq030} is equivalent to $\varphi'(r) > 0$. Hence it is clear that the energy is bounded from below and that 
\begin{equation}\nonumber
 \textrm{ any minimizing sequence is bounded in } W^{1,p} \cap L^\infty(R_{in},R_{out}).
\end{equation}
Therefore we can use direct methods of Calculus of Variations to obtain a minimizer $v$ for this convex problem. 
\end{proof}

In the following, we keep $T_{out}$ fixed and treat $T_{in}$ as a parameter, and write 
$v(r;T)$ for the minimizer of~\eqref{eq13} with $T_{in}=T$.
We also define an interval \[\mathcal{I}:= (T_{out},T_{out}R_{out}/R_{in}).\]

Now we prove a bound on $v'$, which will be useful afterward in showing the uniqueness and some properties of a minimizer $v$:

\begin{lemma}\label{lm-pos}
 Under the assumptions of Theorem~\ref{thm-rel} there exist constants $1 < V_{min} < V_{max}$, such that for any $T \in \mathcal{I}$ the minimizer $v(\cdot;T)$ satisfies 
 \begin{equation}\nonumber
  V_{min} \le v'(r;T) \le V_{max} \textrm{~ for ~} r \in (R_{in},R_{out}).
 \end{equation}
\end{lemma}

\begin{proof}
 Let $T \in \mathcal{I}$ be fixed and let $v(r) := v(r;T)$. We write $\sigma_r$ and $\sigma_\theta$ for the corresponding radial and hoop stress, respectively:
\begin{equation}\label{eq020}
 \sigma_r = \partial_1 f_r(v', v/r),\qquad
 \sigma_\theta = \partial_2 f_r(v', v/r).
\end{equation} 
\begin{enumerate}
 \item
 Since $v$ is minimizer of~\eqref{eq13}, it satisfies Euler-Lagrange equation:
 \begin{equation}\label{eq021}
  (r\sigma_r(r))' = \sigma_\theta \ge 0,\quad \textrm{i.e.}\quad \sigma_r' = \frac{1}{r}\left( \sigma_\theta - \sigma_r \right),
 \end{equation}
 where the inequality $\sigma_\theta \ge 0$ follows from the definition of $f_r$. Therefore $r\sigma_r(r)$ is non-decreasing, and thus we obtain
 \begin{equation}\label{eq025}
  \partial_1f_r(v'(r),v(r)/r) = \sigma_r(r) \ge \frac{R_{in}}{r}\sigma_r(R_{in}) = \frac{R_{in}}{r}T_{in} \ge \frac{R_{in}}{R_{out}}T_{in} \ge \frac{R_{in}}{R_{out}}T_{out} > 0.
 \end{equation}
 We see from Remark~\ref{rmk-f} that
 \begin{equation}\nonumber
  v'(r) > 0  \textrm{~~for } r \in (R_{in},R_{out}).
 \end{equation}
 
 We define
 \begin{equation}\nonumber
  H(\lambda) := \partial_1 f(\lambda,\o(\lambda)) \quad \textrm{ for } \lambda > 1.
 \end{equation}
 This quantity represents force required to uniaxially stretch an elastic body to $\lambda$ times of its original length. It is natural to expect monotonicity of $H$. Under our assumptions this is true. Indeed, \eqref{eq035} implies
 \begin{equation}\nonumber
  H'(\lambda) = (\partial_1 f(\lambda,\o(\lambda)))' = \partial_{11} f(\lambda,\o(\lambda)) + \partial_{12} f(\lambda,\o(\lambda))\o'(\lambda) > 0
 \end{equation}
 and so 
 \begin{equation}\nonumber
  H \textrm{ is a strictly increasing function}.
 \end{equation}
 
Using~\eqref{eq038} and the monotonicity of $r\sigma_r(r)$, $H(v'(r))$ satisfies
\begin{equation}\nonumber
 H(v'(r)) \le \partial_1 f_r(v'(r),v(r)/r) = \sigma_r(r) \le T_{out}R_{out}/R_{in}.
\end{equation}
So by monotonicity of $H$
\begin{equation}\nonumber
 v'(r) \le H^{-1}(T_{out}R_{out}/R_{in}) =: V_{max}.
\end{equation}
 
 \item \label{pt02} We want to show that $\sigma_r > \sigma_\theta$ in $(R_{in},R_{out})$. First, let us prove that $\sigma_\theta \neq \sigma_r$ everywhere. Otherwise let $r_0$ be such that $\sigma_r(r_0) = \sigma_\theta(r_0)$. By differentiating~\eqref{eq020} we obtain
\begin{equation}\nonumber  
 \sigma_r' = \partial_{11}f_r\cdot v'' + \partial_{12}f_r \cdot \frac{1}{r}\left(v' - \frac{v}{r}\right).
\end{equation}
Then using Euler-Lagrange equation~\eqref{eq021} we see that $v$ is a solution to the second order ODE
\begin{equation*}
 \frac{1}{r}\left( \partial_2f_r(v',v/r) - \partial_1f_r( v',v/r)\right) = \partial_{11}f_r(v',v/r)\cdot v'' + \partial_{12}f_r(v',v/r) \cdot \frac{1}{r}\left(v' - \frac{v}{r}\right)
\end{equation*}
with values $v'(r_0) = v(r_0) / r_0 =: \kappa$ for some $\kappa$. At the same time we see that $V(r) := \kappa r$ is a~solution to the same ODE with $V'(r_0) = v'(r_0)$ and $V(r_0) = v(r_0)$. Since $\partial_{11}f_r > 0$ the ODE satisfies the uniqueness principle and so $v = V$ in $(R_{in},R_{out})$, a contradiction with the values of $\sigma_r$ at $R_{in}$ and $R_{out}$.

Since $\sigma_r \neq \sigma_\theta$, we have either $\sigma_\theta > \sigma_r$ or $\sigma_\theta < \sigma_r$ in the whole interval. 
In the first case~\eqref{eq021} would imply that $\sigma_r$ is a non-decreasing function of $r$, 
a contradiction with the boundary conditions for $\sigma_r$. Therefore 
\begin{equation}\label{eq022}
 \sigma_r > \sigma_\theta \textrm{ in }(R_{in},R_{out}).
\end{equation}

 \item By virtue of~\eqref{eq034} and~\eqref{eq022} we see 
\begin{equation}\nonumber  
 v'(r) > v(r)/r\textrm{ for } r\in (R_{in},R_{out}).
\end{equation}
Then it follows from~\eqref{eq038} and~\eqref{eq025} that
\begin{equation}\nonumber
 \partial_1 f_r(v'(r),v'(r)) \ge \partial_1 f_r(v'(r),v(r)/r) \ge \frac{R_{in}}{R_{out}}T_{out} > 0
\end{equation}
and immediately
\begin{equation}\nonumber
 v'(r) \ge V_{min} > 1,
\end{equation}
where $V_{min}$ depends only on $f_r$ and $R_{in}T_{out}/R_{out}$.  
\end{enumerate} 
\end{proof}

We have seen in~\eqref{eq033} that the relaxed density $f_r$ can be expressed in terms of $f$. As a consequence, $f_r$ partially inherits the strict convexity of $f$. 
We use the convexity to show uniqueness of 
$v$:

\begin{lemma}\label{rel_uniqueness}
 Under the hypotheses of Theorem~\ref{thm-rel}, a minimizer $v$ of~\eqref{eq13} is unique.
\end{lemma}

\begin{proof}
Let $v$ be a minimizer of~\eqref{eq13}. By Lemma~\ref{lm-pos} we know that
\[ 1 < V_{min} \le v'(r) \le V_{max} \textrm{ for } r \in (R_{in},R_{out}). \]
From~\eqref{eq032} and \eqref{eq035a} we have that $\partial_{11}f_r(\ll) > 0$ for $\lambda_1 > 1$, which together with convexity of $f_r$ in both variables implies the uniqueness of $v'$. 
Moreover, we know by Remark~\ref{rmk-nonrelaxed} that there exists $r_0 \in (R_{in},R_{out})$ with nontrivial hoop stress:
\begin{equation*}
 \sigma_\theta(r_0) > 0, \textrm{ i.e. } v(r_0)/r_0 > \o(v'(r_0)).
\end{equation*}
Since $f_r(\ll) = f(\ll)$ for $\lambda_1 > 1, \lambda_2 > \o(\lambda_1)$, strict convexity of $f$ (in particular the fact $\partial_{22}f > 0$) implies that $v(r_0)$ is uniquely determined. This together with the uniqueness of $v'$ completes the argument.
\end{proof}


\begin{lemma}\label{lm2.2}
 Let us assume that there exists a non-empty relaxed region, i.e. there is a maximal interval $(A,B) \subset (R_{in},R_{out})$ such that
 \[ v(r)/r \le \o(v'(r)) \textrm { for } r \in (A,B).\]
 Then under the assumptions of Theorem~\ref{thm-rel} we have $A=R_{in}$. 
\end{lemma}

\begin{proof}
 
 
 To prove that $A = R_{in}$, let us assume that $A > R_{in}$. Then $\sigma_\theta(A)=0$ and $\sigma_\theta$ is strictly positive in a left neighborhood $\mathcal{U}$ of $A$. This in particular means that $f(v',v/r) = f_r(v',v/r)$ in $\mathcal{U}$. Hence we can do all our computations with $f$ instead of $f_r$.  
 
 We differentiate~\eqref{eq020} in $\mathcal{U}$ to obtain
\begin{align*}
 \sigma_\theta' = \partial_{12}f\cdot v'' + \frac{1}{r}\partial_{22} f \left(v' - \frac{v}{r}\right) = \frac{1}{r \partial_{11}f} \left( \det D^2f \left(v' - \frac{v}{r}\right) + \partial_{12}f\cdot \sigma_r' r \right),
\end{align*}
where we have used $\sigma_r' = \partial_{11}f(v',v/r)v'' + \partial_{12}f(v',v/r)(v/r)'$ to express $v''$. Now consider the limit $r \nearrow A$. Since $\sigma_\theta(r) \to 0$, from the Euler-Lagrange equation 
$r\sigma_r' = \sigma_\theta - \sigma_r$ we know that $\sigma_r'(r) r \to -\sigma_r(A) = -\partial_1f(v'(A),v(A)/A)$. Therefore we get
\begin{equation*}
\sigma_\theta'(A) = \frac{1}{A \partial_{11}f} \left( \det D^2f \left(v'(A) - \frac{v(A)}{A}\right) - \partial_{12}f\cdot \partial_1f \right) > 0,
\end{equation*}
where the last inequality follows from~\eqref{eq18} with $\lambda_1 = v'(A)$ (note that by the definition of $A$, $w(v'(A)) = v(A) / A$). Therefore there exists no such point $A$ and the relaxed region has to start at $R_{in}$.
\end{proof}

The next lemma shows that there exists a range of loads $T$ such that $v(r;T) > 0$ for $r \in (R_{in},R_{out})$ and $v(\cdot;T)$ has a non-empty relaxed region. 

\begin{lemma}\label{lm2.1}
 Under the assumptions of Theorem~\ref{thm-rel} there exists a range of loads $T$ in $\mathcal{I}$ such that 
 \begin{equation}\label{eq19}
 \frac{v(R_{in};T)}{R_{in}} < \o(v'(R_{in};T)) \textrm{ and } v(R_{in};T) > 0.
\end{equation}
\end{lemma}

\begin{proof}

By Lemma~\ref{lm-pos} and~\eqref{eq027} we know that $\o(v'(R_{in};T)) \ge \o(V_{max})$, and so it is enough to show that
\begin{equation}\nonumber
 0 < v(R_{in};T) < \o(V_{max})R_{in}
\end{equation}
for some range of loads $T$. 

Let $T \in \mathcal{I}$ be arbitrary. As a preliminary step, we observe that $T \mapsto v(R_{in};T)$ is continuous function for $T \in \mathcal{I} \cup T_{out}$. This follows in the standard way from the uniqueness of minimizers. 
We also see that $T \mapsto v'(R_{in};T)$ is a continuous function of $T \in \mathcal{I} \cup T_{out}$. This is an immediate consequence of the optimality condition 
$\partial_1 f_r(v'(R_{in};T),v(R_{in};T)/R_{in}) = \sigma_r(R_{in}) = T$ and monotonicity of 
$\partial_1 f_r(\lambda_1,\lambda_2)$ in the first variable. Then using~\eqref{eq027} we see that $T \mapsto \o(v'(R_{in};T))$ is also a continuous function. By the same argument the same is true for $v'(R_{out};T)$ and $\o(v'(R_{out};T))$ as well. 

\smallskip

Now we turn to the main point: the value of $v(R_{in};\cdot)$ at the endpoints of $\mathcal{I}$. 
If $T = T_{out}$, the solution has the form $v(r;T_{out}) = \kappa r$ for some $\kappa > 1$ such that $\partial_1 f_r(\kappa,\kappa) = T_{in} = T_{out}$, and so clearly $v(R_{in};T_{out}) = \kappa R_{in} > 0$ in this case. In fact, there is no relaxed region (and therefore no compression in the hoop direction) when $T=T_{out}$.

At the other endpoint $T = T_{out}R_{out}/R_{in}$ we no longer have uniqueness of a minimizer for~\eqref{eq13} (see Remark~\ref{rmk-rel}). Nevertheless, such minimizer is unique up to a translation, i.e. $v'(\cdot;T)$ is uniquely determined, and it has to satisfy $v(R_{out};T)/R_{out} \le \o(v'(R_{out};T))$. Indeed, since $(r\sigma_r)' = \sigma_\theta \ge 0$ and $R_{in}\sigma_r(R_{in}) = R_{out}\sigma_r(R_{out})$, we have that $\sigma_\theta \equiv 0$. This immediately implies $v(R_{out};T)/R_{out} \le \o(v'(R_{out};T))$.

Now we consider a sequence of loads $T_k \in \mathcal{I}$, $T_k \nearrow T = T_{out}R_{out}/R_{in}$. As before we have a sequence of unique minimizers $v_k = v(\cdot;T_k)$, where each of them has a non-empty non-relaxed region (see Remark~\ref{rmk-nonrelaxed}). We have shown in Lemma~\ref{lm2.2} that any relaxed region has to start at $R_{in}$, and therefore having a~non-relaxed region implies $v_k(R_{out})/R_{out} > \o(v_k'(R_{out}))$.
We also know that $v(\cdot;T)$ is determined only up to an additive constant, and so theoretically it is possible that different subsequences of $\{ v_k \}$ are converging to different minimizers $v(\cdot;T)$. We show that this is not the case, i.e. that $v_k$ converges to one particular minimizer $v(\cdot;T)$. Let us take a subsequence of $\{v_k\}$ (labeled the same) which converges to some $\tilde v := v(\cdot;T)$. Since $v_k(R_{out})/R_{out} > \o(v_k'(R_{out}))$, continuity of $v(R_{out};\cdot)$ and $\o(v'(R_{out};\cdot))$ implies \[\tilde v(R_{out})/R_{out} \ge \o(\tilde v'(R_{out})).\]
At the same time, we know that $\tilde v$ doesn't have a non-relaxed region, and so $\tilde v(R_{out})/R_{out} \le \o(\tilde v'(R_{out}))$. This shows that $\tilde v(R_{out}) = \o(\tilde v'(R_{out}))R_{out}$, i.e. the limiting $v(\cdot;T)$ is uniquely determined. Therefore the whole sequence $\{ v_k \}$ converges to this particular minimizer. This in particular shows that $v(R_{in};\cdot)$ is continuous from the left at $T = T_{out}R_{out}/R_{in}$. 

To summarize, we have shown that 
\begin{equation}\nonumber
\begin{array}{ll}
 v(R_{in};T) > 0 \qquad\qquad & \textrm{ for } T = \textrm {the left endpoint of } \mathcal{I}; \\
 v(R_{in};T) < \o(V_{max})R_{in} & \textrm{ for } T = \textrm {the right endpoint of } \mathcal{I}.
\end{array}
\end{equation}
Continuity of $v(R_{in};\cdot)$ and the fact that $0 < \o(V_{max})$ implies that there are some $T \in \mathcal{I}$ such that 
\[ 0 < v(R_{in};T) < \o(V_{max})R_{in}.\]
This completes the proof of the lemma. 
\end{proof}

\begin{cor*}
 Let $T \in \mathcal{I}$ be such that~\eqref{eq19} is true. Then under the assumptions of Theorem~\ref{thm-rel} there exists $L \in (R_{in},R_{out})$ such that the interval $(R_{in},R_{out})$ splits into a relaxed region $(R_{in},L)$  and a non-relaxed region $(L,R_{out})$. 
\end{cor*}

\begin{proof}
 Since $T$ satisfies~\eqref{eq19}, the relaxed region is non-empty. Then by Lemma~\ref{lm2.2} it has to be of the form $(R_{in},L)$ for $L \in (R_{in},R_{out}]$. Using Remark~\ref{rmk-nonrelaxed} we see that the non-relaxed region is also non-empty, and so necessarily $L < R_{out}$. This completes the proof. 
\end{proof}

It remains to show~\eqref{eq-th2}. Let us compute the derivative of $h$ for $r < L$:
\begin{equation}\label{eq-h}
h'(r) = (\o)'(v'(r))\cdot v''(r) - \frac{1}{r}\left(v'(r)-\frac{v(r)}{r}\right).
\end{equation}
Since $r < L$, the hoop stress $\sigma_\theta = 0$ and together with Euler-Lagrange equation ($0=(r\sigma_r)'$ in this case) we obtain
\begin{align*}
 0 = \sigma_\theta(r) = \partial_2f (v'(r), \o(v'(r))),\qquad \
 \frac{\alpha}{r} = \sigma_r(r) = \partial_1f (v'(r), \o(v'(r))),
\end{align*}
($\alpha$ is a positive constant) and by differentiating 
\begin{align*}
 0 &= \partial_{12}f ( v'(r), \o(v'(r))) v''(r) + \partial_{22}f (v'(r), \o(v'(r)))\cdot(\o)'(v'(r))\cdot v''(r),\\
-\frac{\alpha}{r^2} &= \partial_{11}f (v'(r), \o(v'(r))) v''(r) + \partial_{12}f (v'(r), \o(v'(r)))\cdot(\o)'(v'(r))\cdot v''(r).
\end{align*}
We solve this linear system for $v''(r)$ and $(\o)'(v'(r)) v''(r)$ and substitute the result into~\eqref{eq-h}.
At $r=L$ we have $v' - v/r = v' - \o(v')$, thus $h'(L) < 0$ is equivalent to
\begin{equation}\label{eqgrowth}
\det D^2f(v'(r),\o(v'(r))) \cdot \left(v'(r) - \o(v'(r))\right) >  \ \partial_{12}f(v'(r),\o(v'(r))) \cdot \partial_1f(v'(r),\o(v'(r))), 
\end{equation}
exactly matching condition~\eqref{eq18}. 

This completes the proof of Theorem~\ref{thm-rel}.

\subsection{Uniqueness of the minimizer}

In this section we show the 
uniqueness of the minimizer for the relaxed problem~\eqref{eq4}. 

First, it is easy to check that $u_0$ defined in~\eqref{eq3} is a minimizer for the relaxed energy $E_0$ (since the functional is convex, any solution of the Euler-Lagrange
equations is a minimizer).
Since $E_0$ is not strictly convex, the uniqueness of the minimizer is not clear. It is however true: we now show that $u_0$ is (up to an additive constant) the only minimizer of the relaxed problem.
Indeed, suppose that there exists another minimizer $u_1$ of the relaxed problem.
By convex duality for $E_0$ we have:
\begin{equation}\nonumber
 \min_{u \in W^{1,p}(\Omega,\mathbb{R}^3)} E_0(u) = \
\max_{ \substack{\sigma \in L^{p'}(\Omega,\mathbb{R}^{3\times 3}) \\ \mathrm{div} \sigma = 0,\ \sigma . n = T\,  \mathrm{ at }\, \partial\Omega}} \int_{\Omega} -D(\sigma) \dxy,
\end{equation}
where $D(\sigma) := \sup_{F \in \R^{3\times3}} \left<\sigma, F\right> - \wr(F)$ is the convex conjugate of $\wr$. Since $u_0$ is a minimizer of $E_0$, we know that the maximum on the RHS is attained for $\sigma_0 = \frac{\partial \wr}{\partial F}(\grad u_0)$. From the definition of the convex conjugate $D$ we see that 
$\left< \sigma_0 , \grad u_1\right> - \wr(\grad u_1) \le D(\sigma_0)$, and after integration $E_0(u_1) \ge \int_\Omega -D(\sigma) \dxy$.
Since the deformation $u_1$ is a minimizer of $E_0$ as well, we obtain an equality in the last relation. Hence 
$\left< \sigma_0 , \grad u_1\right> - \wr(\grad u_1) = D(\sigma_0)$ a.e. in $\Omega$ and consequently $\sigma_0 = \frac{\partial \wr}{\partial F}(\grad u_1)$ a.e. in $\Omega$. 
It follows that 
\begin{equation}\label{equalgrad}
 \grad u_0 = \grad u_1
\end{equation}
at points where $\wr$ is strictly convex at $Du_0$ (i.e. at points where the eigenvalues $\lambda_1 \ge \lambda_2$ of $(Du_0^TDu_0)^{1/2}$ satisfy $\lambda_1 > 1$ and $\lambda_2 > \o(\lambda_1)$). 


We have proved in Theorem~\ref{thm-rel} that the hoop stress is tensile exactly in the non-relaxed region 
\begin{equation}\label{eq-relaxed}
\nonrelaxed := \left\{ x : L < |x| \le R_{out}\right\} 
\end{equation}
with $R_{in} < L < R_{out}$. By $\relaxed = \Omega \setminus \nonrelaxed$ we denote its complement, i.e. the relaxed region. Since the both stresses are tensile in $\nonrelaxed$, as a consequence of~\eqref{eq033} we have $\wr(\grad u_0) = \w(\grad u_0)$.

Using the strict convexity of $\wr$ in the non-relaxed region $\nonrelaxed$ we see that 
\begin{equation}\label{u0=u1}
 u_0(x) = u_1(x) + C
\end{equation}
for $x \in \nonrelaxed$. We may assume without loss of generality that $C = 0$. In the relaxed region $\relaxed$ we need to replace~\eqref{equalgrad} by
\begin{equation*}
 \grad u_0(x) \cdot n(x) = \grad u_1(x) \cdot n(x),\qquad x\in \relaxed
\end{equation*}
where $n(x) = \frac{x}{|x|}$ is a unit vector in the radial direction. This concludes the proof of uniqueness, because integrating the last relation in the radial direction extends validity of~\eqref{u0=u1} to the whole set $\Omega$.

\section{The 2D result}\label{sec-2d}

In this section we show the matching upper and lower bounds in the simplified Kirchhoff-Love setting. After stating the main result of this section we prove the upper bound by superimposing wrinkles on the solution $u_0$ of the relaxed problem. We use Lemma~\ref{lm-gamma1} to create simple wrinkles whereas Lemma~\ref{lm-gamma2} provides a tool to create a family of wrinkles with changing wavelength near the free boundary. In the rest of the section we prove a matching lower bound. Using Lemmas~\ref{lm1}, \ref{lm2},  and~\ref{lm3} we show that if $E_h(u_h)$ is close to $\mathcal{E}_0$, then also $\grad u_h$ has to be close to $\grad u_0$. As a consequence we obtain a bound on the out-of-plane displacement $u_{h,3}$. By interpolation we show the smallness of $\grad u_{h,3}$, which allows us to project $u_h$ to the plane without changing its energy too much. The final ingredient is a comparison of $u_0$ with the projection of $u_h$ (Lemma~\ref{lm9}).

For Theorem~\ref{thm-2d} we will assume that the lower bound in~\eqref{eq032} holds with $p=2$ rather than just $1 < p \le 2$. 
The stronger assumption $p=2$ is only required for the second half of the proof of the lower bound (e.g. for the interpolation), whereas the first half of the proof of the lower bound (especially Lemmas~\ref{lm1}, \ref{lm2}, \ref{lm3}, and the Poincar\'e inequality) requires only $1 < p \le 2$. Since the real purpose of the two-dimensional result is to lay the ground for the proof of the three-dimensional case, 
many of the preparatory lemmas are proved in the more general setting $1 < p \le 2$.

Let us now state the main result of this section:

\begin{theorem}\label{thm-2d}
Let us assume \eqref{eq032} with $p=2$ 
and all hypotheses of Theorem~\ref{thm-rel}. 
Then there exist constants $0 < C_1 < C_2$ independent of $h$ such that
\begin{equation}\label{eq-2d}
 \mathcal{E}_0 + C_1h \le \min_{u} E_h(u) \le \mathcal{E}_0 + C_2h,
\end{equation}
where $\mathcal{E}_0$ is the minimal value of the relaxed problem~\eqref{eq4}.
\end{theorem}

\begin{remark*}
Theorem~\ref{thm-2d} doesn't necessarily require all hypotheses of Theorem~\ref{thm-rel}. In fact, some of them can be replaced by assumptions on the solution $u_0$ of the relaxed problem (which are consequences of Theorem~\ref{thm-rel}).
\end{remark*}

\subsection{The upper bound in the two-dimensional setting}\label{sect-ub2d}

To obtain the upper bound, we must construct a test function $u_h$ for any (small) $h > 0$ with energy 
\begin{equation*}
 E_h(u_h) \le \mathcal{E}_0 + Ch,
\end{equation*}
where the constant $C$ is independent of $h$. A naive approach would be to superimpose a ``single family of wrinkles'' (with a well-chosen period independent of $r$, and a well-chosen amplitude that depends on $r$), c.f.~\cite{bib-benny1} (see also \cite{epstein-wrinkles} for a similar calculation). The energy associated with this $u_h$ has the expected scaling 
away from $r=L$ (the edge of the wrinkled region). However the membrane and bending energies are both singular at $r=L$; as a result, the total energy (after integration) is too large, of order $\mathcal{E}_0 + O(h |\log h|)$. The singularity in the bending term can be avoided by introducing a boundary layer, however we have not found a similar way of avoiding the singularity in the membrane term (which is associated with stretching in the radial direction).

To get a linear correction a more complicated construction seems necessary, using a ``cascade of wrinkles'' rather than a ``single family of wrinkles.'' In other words, the period of the wrinkling changes repeatedly as one approaches the edge of the wrinkled region. Constructions of this type have been used in other settings, for example in studies of compressed thin film blisters~\cite{bib-blisters,bib-jin-sternberg}.

Recall that the solution of the relaxed problem has compressive hoop strain when $r<L$. The essential purpose of the wrinkling is to avoid this compressive hoop strain by out-of-plane buckling. In the following lemmas, we write $\eps(r)$ for the compressive hoop strain to be avoided. Up to a factor of $2\pi r$, this amounts to ``the amount of arclength to be wasted by wrinkling'' along the image of the circle of radius $r$.

\begin{lemma}[see Lemma 2 in \cite{bib-blisters}]\label{lm-gamma1}
For every $\eps > 0$ there exists a smooth $\mathcal{C}^\infty$ 
planar curve $\gamma(\eps) = (\gamma_1(\eps)(t),\gamma_2(\eps)(t)) : \R \to \R^2$ with properties
\begin{gather*}
|\partial_t\gamma| = 1 + \eps, \quad \partial_t\gamma_1\ge 0, \quad \gamma(-t) = -\gamma(t) \\
\gamma(t+2\pi) = \gamma(t) + \begin{pmatrix}2\pi\\0\end{pmatrix},
\end{gather*}
and satisfying the bound
\begin{equation}
\begin{gathered}\nonumber
|\gamma_1-t|+|\partial_t\gamma_1-1|+|\partial_{tt}\gamma_1|\le C\eps,\qquad |\gamma_2|+|\partial_t\gamma_2|+|\partial_{tt}\gamma_2| \le C \epse^{1/2},\\
|\partial_\eps \gamma_1| \le C, \quad
|\partial_\eps \gamma_2| \le C \epse^{-1/2}, \quad
|\partial^2_\eps \gamma_2| \le C \epse^{-3/2},
\end{gathered}
\end{equation}
where $C$ does not depend on $\eps$. Moreover, the bound is sharp (in terms of the scaling in $\eps$) for small values of $\eps$.
\end{lemma}

\begin{proof}[Idea of the proof (see Lemma~2 in \cite{bib-blisters} for more detail):]
 We construct $\gamma$ by re\-pa\-ra\-me\-tri\-zing the cur\-ves
\begin{equation*}
  \tilde\gamma(t) = \left(\begin{array}{c} t \\ A \sin t \end{array} \right)
\end{equation*}
so that $|\partial_t\gamma|=1+\eps$, where $A$ is chosen such that $\int_0^{2\pi} |\partial_t\tilde\gamma(t)| \ud t = 2\pi(1+\eps)$. By considering the small-$\eps$ limit we obtain $A = \rho(\epse)\epse^{1/2}$ ($\rho$ being a smooth function on $[0,1]$). This leads easily to the desired estimates.
\end{proof}

We would like to use Lemma~\ref{lm-gamma1} to superimpose wrinkles on top of the planar deformation $u_0$ obtained from the solution $v$ to the relaxed problem~\eqref{eq13}. Though this naive construction does not achieve optimal energy scaling, it can be modified (using Lemma~\ref{lm-gamma2}) to obtain a construction with optimal energy scaling. Therefore it makes sense to analyze the naive construction, to understand why it fails and also to motivate the successful construction.


The naive construction is based on Lemma~\ref{lm-gamma1} and proceeds as follows. After we obtain the parameter $\eps(r)$ (amount of wastage of arclength) from $v$, we determine the right period of wrinkles to obtain the optimal scaling. A bit of calculation reveals that the stretching and bending terms are of the same order when the 
number of wrinkles is of order 
\begin{equation}\nonumber
k := [h^{-1/2}],
\end{equation}
where the brackets denote an integer part. From $v$ we obtain the amount of wastage of arclength as
\begin{equation}\nonumber
 \eps(r) := \frac{w(v'(r))}{v(r)/r} - 1
\end{equation}
for $r \in (R_{in},L)$.
Following the proposed idea we define  
a solution $\bar u_h$ in $\Omega_R$ (using radial coordinates $r,\theta$) by
\begin{equation}\label{eq-naive}
 \bar u_h(r,\theta) = v(r) \hat r + \frac{\gamma_1(\eps(r))(k\theta)}{k} \hat \theta + v(r) \frac{\gamma_2(\eps(r))(k\theta)}{k} e_3
\end{equation}
and as $u_0$ elsewhere.
From~\eqref{eq-th2} we know that $\eps \sim L-r$ (up to a factor) for $0 < L - r << 1$. Since $\partial_\eps^2 \gamma_2 \sim \epse^{-3/2}$ by Lemma~\ref{lm-gamma1}, the contribution of $\partial_{rr}\bar u_{h,3}$ to the bending energy is divergent:
\begin{equation*}
  h^2 \int_{R_{in}}^L \left|\partial_{rr}\bar u_{h,3}\right|^2 r \dr \sim	 h^2 k^{-2} \int_{R_{in}}^L \left((L-r)^{-3/2}\right)^2 r \dr = \
h^3 \int_{R_{in}}^L (L-r)^{-3} r \dr. 
\end{equation*}
It is clear that the bending energy is of order $h$ in the region $R_{in} < r < L-h$.
Therefore a boundary layer in the region $L-h < r < L$ would solve this issue provided 
the integral over this layer of the new $|\partial_{rr} \bar u_{h,3}|^2$ is at most of order $h^{-1}$.

Now let us try to compute the contribution from the stretching energy near the transition from the relaxed to the non-relaxed state. In the region $R_{in} < r < L-\delta$ (for a fixed $\delta > 0$) the term $\partial_r \bar u_{h,3}$ from the membrane energy satisfies
\begin{equation*}
\int_{R_{in}}^{L-\delta} |\partial_r \bar u_{h,3}|^2 r \dr \approx \int_{R_{in}}^{L-\delta} \
k^{-2} \left( \left|v'(r) \gamma_2\right|^2 + \left|v(r)\partial_\eps\gamma_2\right|^2\right) r \dr.
\end{equation*}
By Lemma~\ref{lm-gamma1} the first term in the parentheses is of order $\eps$ whereas the second term is of order $\epse^{-1}$. Therefore,  near the free boundary (where $\eps << 1$) the second term is dominant and we obtain 
\begin{equation*}
\int_{R_{in}}^{L-\delta} |\partial_r \bar u_{h,3}|^2 r \dr \approx \int_{R_{in}}^{L-\delta} k^{-2}\epse^{-1} r \dr \approx h(\log (\delta) - \log(L-R_{in})).
\end{equation*}
We see that by setting $\delta = h$ (or any power of $h$) we would obtain energy scaling $h |\log h|$.



We now begin discussion of the successful construction, which uses a ``cascade of wrinkles'' near $r = L$. The main tool is Lemma~\ref{lm-gamma2}. Whereas Lemma~\ref{lm-gamma1} involved wrinkled curves, Lemma~\ref{lm-gamma2} involves wrinkled strips, in which the length scale of wrinkling doubles from one side to the other. It provides the basic building block for our cascade of wrinkles. 

\begin{lemma}\label{lm-gamma2}
Let $B = (0,l) \times (0,w)$ with $0 < l,w \le  1$, $F : (0,l) \to \R$, and $e$ be a positive 
function on $(0,l)$ satisfying $|e'| \le  c$, $\left|e''\right| \le c$, and $l/c \le e \le c l$ for some $c > 0$. Then there exists a smooth deformation $\Psi(s,t)$ defined on $B$ and $w$-periodic in the $t$ variable such that for any $t \in (0,w)$ the following holds:
\begin{gather*}
 \Psi_1(s,t) = F(s),\quad s \in (0,l)\\
 \left(\Psi_2,\Psi_3\right)(s,t) = w \gamma(e(s))(\frac{t}{w}), \quad s \in (0,l/4) \\
 \left(\Psi_2,\Psi_3\right)(s,t) = \frac{w}{2}\gamma\left(e(s)\right)\left(\frac{t}{w/2}\right), \quad s \in (3/4\,l,l) \\
 \Psi(s,0) = (F(s),0,0), \quad \Psi(s,w) = (F(s),w,0), \quad s \in (0,l),\\
\end{gather*}
and 
\begin{gather*}
  \left|\partial_s \Psi_2\right|^2 + |\partial_s \Psi_3|^2 \le Cw^2l^{-1}, \quad \left|\partial_t \Psi\right| = 1 + e(s),\quad
  \left|\grad^2 \Psi\right|^2 \le C\left(l^{-3}w^2 + w^{-2}l\right),
\end{gather*} 
where $C$ depends just on $c$ and $\gamma$ is the curve defined in Lemma~\ref{lm-gamma1}.
\end{lemma}

\begin{proof}
To prove this lemma we just need to define $\Psi_2, \Psi_3$ such that the required estimates are true. The idea of the construction is very similar to the proof of Lemma~\ref{lm-gamma1}. To simplify the notation we first assume $w = 1$.

First, let us fix 
$0 < \varepsilon < 1$ and $\alpha \in [0,1]$. We consider a planar curve:
\begin{equation*}
  \tilde \gamma : t \mapsto \left(t,A\left[ (1-\alpha)\sin \left(2\pi t  \right) + \alpha \sin\left(4\pi t\right)\right] \right),
\end{equation*}
where $A=A(\varepsilon,\alpha)$ is such that the length of $\tilde\gamma([0,w])$ is exactly $(1+\varepsilon)$. More specifically, we define 
\begin{equation*}
 \Lambda(t) = \frac{1}{1+\varepsilon} \int_0^t |\tilde\gamma'(\tau)| \ud \tau \
= \frac{1}{1+\varepsilon} \int_0^t \sqrt{1 + (2\pi A)^2\left[ (1-\alpha)\cos (2\pi \tau) + 2\alpha \cos (4\pi \tau)\right]^2} \ud \tau
\end{equation*}
with $A$ such that $\Lambda(1)=1$ (since $\Lambda$ is strictly increasing function of $A$, there exists a unique such $A$).

Considering the small-$\varepsilon$ limit, we obtain 
\begin{equation}\nonumber
 A = \rho(\varepsilon,\alpha) \varepsilon^{1/2},
\end{equation}
where $\rho \in \mathcal{C}^\infty([0,1]\times[0,1])$. Consequently we have that 
\begin{equation}\nonumber
\left|\frac{\partial A}{\partial^k \varepsilon \partial^l \alpha}\right| \le C \varepsilon^{1/2 - k}. 
\end{equation}
Using $\tilde \gamma$ we define a new reparametrized curve
\begin{equation}\nonumber
 \Gamma(\varepsilon,\alpha,t) = \tilde \gamma \circ \Lambda^{-1}(t).
\end{equation}
This curve obviously satisfies $|\Gamma'(t)| = 1+\varepsilon$ and 
\begin{equation*}
 \Gamma_1(\varepsilon,\alpha,t) = t - \varepsilon \rho_1(\varepsilon,\alpha,t),\qquad \
 \Gamma_2(\varepsilon,\alpha,t) = \varepsilon^{1/2} \rho_2(\varepsilon,\alpha,t),
\end{equation*}
where $\rho_1, \rho_2 \in \mathcal{C}^\infty([0,1]\times[0,1]\times \R)$.
From there we get estimates 
\begin{gather*}
|\partial_{\varepsilon} \Gamma_1| + |\partial_{t} \Gamma_1| \le C, \qquad |\partial_{\alpha} \Gamma_1| \le C\varepsilon,\\
|\partial_{\varepsilon \varepsilon} \Gamma_1| + |\partial_{\varepsilon \alpha} \Gamma_1| + |\partial_{\varepsilon t} \Gamma_1| \le C, \qquad
|\partial_{\alpha \alpha} \Gamma_1| + |\partial_{\alpha t} \Gamma_1| + |\partial_{tt} \Gamma_1| \le C\varepsilon,\\
|\partial_{\varepsilon}^k\partial_{\alpha}^l\partial_t^n \Gamma_2| \le C\varepsilon^{1/2-k}\quad k,l,n\ge 0.
\end{gather*}
Now we are ready to define the map $\Psi$. We set
\begin{align*}
 \Psi_1(s,t) := F(s), \quad \Psi_2(s,t) := \Gamma_1(e(s),\phi(s),t), \quad \Psi_3(s,t) := \Gamma_2(e(s),\phi(s),t),
\end{align*}
where $\phi$ is a smooth increasing function on $(0,l)$ satisfying
\begin{align*}
 \phi(s) &= 0 \qquad s \in (0,l/4),\\
 0 \le \phi(s) &\le 1 \qquad s \in (l/4,3/4\ l),\\
\phi(s) &= 1 \qquad s \in (3/4\ l,l),
\end{align*}
and $\phi' \le 3/l$. Then
\begin{align*}
 \partial_s \Psi(s,t) &= \left(F'(s),\partial_\varepsilon \Gamma_1 e' + \partial_\alpha \Gamma_1 \phi',\partial_\varepsilon \Gamma_2 e' + \partial_\alpha \Gamma_2 \phi'\right), \\
 \partial_t \Psi(s,t) &= \left(0,\partial_t \Gamma_1,\partial_t \Gamma_2\right),
\end{align*}
and using previous estimates together with $e \approx l$ we obtain desired bounds
\begin{gather*}
  \left|\partial_t \Psi\right| = 1 + e(s),\qquad \ 
  \left|\partial_s \Psi_2\right|^2 + |\partial_s \Psi_3|^2 \le Cl^{-1}.
\end{gather*}
To finish the proof in the case $w=1$, we get the estimates on $\grad^2 \Psi$ in the same way as for the first derivatives of $\Psi$:
\begin{gather*}
|D^2 \Psi_1| \le C  \\ 
|\partial_{tt} \Psi_2| \le Cl,\quad \left| \partial_{st} \Psi_2 \right| \le C, \quad \left| \partial_{ss} \Psi_2 \right| \le Cl^{-1}\\
|\partial_{tt} \Psi_3| \le Cl^{1/2},\quad \left| \partial_{st} \Psi_3 \right| \le Cl^{-1/2}, \quad \left| \partial_{ss} \Psi_3 \right| \le Cl^{-3/2}.
\end{gather*}

It remains to show the lemma for general $w$, i.e. we need to define $\Psi$ on $B = (0,l)\times(0,w)$. 
Let $\widetilde \Psi$ come from the proof of this lemma for the case $w=1$ ($\widetilde \Psi$ is defined on $(0,l)\times(0,1)$).
Then we simply set
\begin{equation}\nonumber
 \Psi(s,t) = \left( F(s), w \widetilde \Psi_2(s,t/w), w \widetilde \Psi_3(s,t/w) \right).
\end{equation}
Now it is an easy calculation to show that $\Psi$ satisfies all the estimates.
\end{proof}

\begin{remark*}
Later we will use Lemma~\ref{lm-gamma2} to create a test function for the upper bound for $E_h$. The definition of the function $F$ from Lemma~\ref{lm-gamma2} will be based on $v$ (the solution of the relaxed problem~\eqref{eq13}), and $1+e$ will be the corresponding natural width. Let $M(s) =  F'(s) e_1 \otimes e_1 + (1+e(s)) e_2 \otimes e_2  $ be a $3 \times 2$ matrix. Since $1+e(s)$ is the natural width, we have that $D\w(M(s)) = a(s) e_1 \otimes e_1$ for some scalar function $a(s)$. By definition we know $M(s)_{11} = \partial_s\Psi_1(s)$. Then the boundedness of $D^2\w$ (see~\eqref{eq032}) implies 
\begin{align*}
\int_B \w(\grad \Psi) - \int_B \w (M(s)) &\le \ 
 C \int_B \left| \left( \partial_s \Psi, \partial_t \Psi\right) - \left(F'(s) e_1 \otimes e_1 + \partial_t\Psi \otimes e_2\right) \right|^2 \\ &= \int_B |\partial_s \Psi_2|^2 + |\partial_s \Psi_3|^2,
\end{align*}
where we used that $|\partial_t\Psi| = 1+e$ (see Lemma~\ref{lm-gamma2}) and the rotational invariance of $\w$. Using Lemma~\ref{lm-gamma2} we obtain
\begin{equation}\nonumber
 \int_B \w(\grad \Psi) + h^2\left| \grad^2 \Psi\right|^2 - \int_B \w (M(s)) \le C\left( w^3 + h^2\left[w^3l^{-2} + w^{-1}l^2\right]\right).
\end{equation}

As already mentioned, we will later set $F := v$, the solution of the relaxed problem, and $e := \epse$, the excess arclength, so that $\int \w(M)$ is the energy of the relaxed solution. This remark will be then used to compare the elastic energy of the constructed deformation with the energy of the relaxed solution (which is $\mathcal{E}_0$).
\end{remark*}

\begin{remark*}
Since $\Psi$ is periodic in $t$, we can assume it is defined in an infinite strip $(0,l) \times \R$.
\end{remark*}

\begin{remark}\label{rmk2}
 In Lemma~\ref{lm-gamma2} we have estimated the size of $\grad^2 \Psi$, not only $\grad^2 \Psi_3$ (in fact, the third component of $\Psi$ was the most oscillatory, and so it is larger than other two). We will use this fact later in the proof of the upper bound in the general three-dimensional setting.
\end{remark}

By~\eqref{eq-th2} and smoothness of $v$ we can choose a small $\delta > 0$ and constants $0 \le c_2 < c_1$ such that $-c_1 \le \partial_r\eps \le -c_2$ in the interval $(L-\delta,L)$. We define a deformation $u_h$ by changing $\bar u_h$ (defined in~\eqref{eq-naive}) in the region $L-\delta < r < L$. The idea is to create a cascade of wrinkles by superimposing wrinkles coming from Lemma~\ref{lm-gamma2} in smaller and smaller rectangles as we approach $r = L$. We define (for a non-negative integer $n$): 
\begin{gather}\nonumber
 \qquad I_n := (L-\delta 4^{-n},L-\delta4^{-(n+1)}),\\ \nonumber
 a_n := L-\delta4^{-n}, \qquad l_n := |I_n|, \qquad w_n := \frac{2\pi}{k} 2^{-n},
\end{gather}
and for $r \in I_n$ we set:
\begin{equation}\nonumber
 u_h(r,\theta) = v(r)\hat r + \Psi_2\hat \theta + v(r) \Psi_3\hat x_3
\end{equation}
where $\Psi$ comes from Lemma~\ref{lm-gamma2} applied to the rectangle $I_n \times (0,w_n)$ with
\begin{gather*}
 F(s) := v(s-a_n), \qquad e(s) := \eps(s-a_n), \qquad \
 w := w_n, \qquad l := l_n.
\end{gather*}
We do this construction in the region $U_N := \bigcup_{n=0}^N I_n \times (0,2\pi)$ with $N = -1/2 \log_2(h)$ ($N$ is chosen such that $2^{2N} = h^{-1}$). Since the doubling of a period defined in Lemma~\ref{lm-gamma2} happens strictly inside the given interval,
the first and second derivatives of $u_h$ are continuous at each $a_n$ (i.e. there are no jumps in the first derivative of $u_h$ between $I_n$ and $I_{n+1}$).

To finish we just need to define $u_h$ in a region close to $r=L$. Consider the strip $(L-\delta 4^{-N},L)$ (observe that it includes $I_N$). By Lemma~\ref{lm-gamma1} the amplitude of $u_{h,3}$ at $a_N$ is of order $w_N \eps^{1/2}(a_N) \approx \delta^{1/2} h^{3/2}$.
The length of the interval $I_N$ is $\delta\left(4^{-N} - 4^{-(N+1)}\right) = 3/4\cdot \delta h$. 
We multiply $u_h$ by a smooth cut-off function in $I_N$ to bring the out-of-plane displacement of $u_h$ to zero at $a_{N+1}$.
Since the length of the interval $I_N$ is of order $h$ and the value of $\grad u_h$ is bounded in that interval, the membrane energy is bounded by $Ch$ in this interval. The second derivative of the new deformation in this region is at most of order $h^{-1/2}$, and so the bending energy in this region is less than $Ch^2 * (h^{-1/2})^2 * h = Ch^2$. 
Overall we obtain:
\begin{align*}
 E_h(u_h) - E_0(u_0) &\le Ch + C\sum_{n=0}^N 2^{n}k \left[ w_n^3 + h^2\left( l_n^{-2}w_n^3 + w_n^{-1}l_n^2\right) \right] \\
 &\le C h \left(1 + h^22^{2N}\right) \le Ch.
\end{align*}

\subsection{The lower bound in the two-dimensional setting}\label{sect-lb2d}

In this section we want to prove the lower bound
\begin{equation}\label{eq76}
 \min_{u \in W^{2,2}(\Omega)} E_h(u) \ge \mathcal{E}_0 + ch
\end{equation}
for some $c > 0$ independent of $h$. Our argument uses the convexity of the relaxed problem; we shall have to work a lot because the relaxed problem is not strictly convex. We will proceed by contradiction, assuming there is a deformation $u$ with energy very close to the energy $\mathcal{E}_0$ of the relaxed solution.  After obtaining a bound on the out-of-plane displacement $u_3$ we use interpolation to show smallness of $\grad u_3$. This allows us to project the deformation $u$ into the $x$-$y$ plane without altering its energy too much (i.e. we obtain a planar deformation with energy close to $\mathcal{E}_0$). Finally, we conclude the proof by showing that it is not possible to have a planar deformations with energy close to $\mathcal{E}_0$. 

The results in the first part of this section will be useful also later for the proof of the three-dimensional case; therefore we prove them assuming only $1 < p \le 2$ in~\eqref{eq032}. On the other hand, it is convenient to assume $p=2$ in the interpolation argument used later in this section, and so from that point on we assume $p=2$.

We define
\begin{equation}\label{gp}
 g_p(t) := \begin{cases} \frac{t^2}{2}, &\mathrm{ if }\ 0 \le t \le 1, \\
								\frac{t^p}{p} + \frac{1}{2} - \frac{1}{p}, & \mathrm { if }\ t > 1, \\
           \end{cases}
\end{equation}
for some $ 1 < p \le 2$. We observe that the function $g_p$ is monotone, convex, and $C^1$. Since it also satisfies $g_p(2t) \le 4 g_p(t)$, convexity of $g_p$ implies
\begin{equation}\label{gp2}
 g_p(a+b) \le 2(g_p(a) + g_p(b)).
\end{equation}

The proof of the lower bound is divided into six steps:

{\bf\noindent Step 1:} To proceed by contradiction, we assume that for any small $0< \delta \le 1$\footnotedva{added $\delta \le 1$ in step 1} there exists a sequence of functions $u_h$ s.t.
\begin{equation}\nonumber
 E_h(u_h) \le \mathcal{E}_0 + \delta h.
\end{equation}
This is equivalent to 
\begin{equation}\label{eq005}
 E_h(u_h) - E_0(u_0) \le \delta h,
\end{equation}
hence using $\wr \le \w$ and definition of $E_0$ we immediately obtain
\begin{equation}\nonumber
 \int_\Omega \wr(\grad u_h) - \wr(\grad u_0) \dxy + B(u_h - u_0) = E_0(u_h) - E_0(u_0) \le \delta h.
\end{equation}
Since $u_0$ is the minimizer of the relaxed energy $E_0$, it has to satisfy Euler-Lagrange equation
\begin{equation}\nonumber
\int_\Omega D\wr(\grad u_0) : \grad\varphi \dxy + B(\varphi) = 0
\end{equation}
for any test function $\varphi \in W^{1,p}(\Omega,\R^3)$. Since $D\wr(\grad u_0)$ is bounded, we can easily take test functions in $W^{1,1}(\Omega,\R^3)$. Using the relation for $\varphi := u_0 - u_h$ yields
\begin{equation}\label{eq84}
 \int_\Omega \wr(\grad u_h) - \wr(\grad u_0) - D\wr(\grad u_0) : (\grad u_h - \grad u_0) \dxy \le \delta h.
\end{equation}

{\noindent\bf Step 2:} We would like to obtain a pointwise lower bound on the integrand of the last relation.

\begin{lemma}\label{lm1}
Assume~\eqref{eq032} holds. Let $F \in \R^{3\times2}$ be an orthogonal matrix with singular values $\lambda_1 > 1$, $\lambda_2 < \lambda_1$, and $n$ be the right singular vector corresponding to $\lambda_1$ (i.e. $F^TF n = \lambda_1^2\, n$). If there exist $\kappa > 0$ and an open neighborhood $\mathcal{U}$ of $(\lambda_1,\lambda_2)$ such that 
\begin{equation}\label{eq85}
 D^2 f_r(\sigma_1,\sigma_2) \ge \kappa e_1 \otimes e_1 \ \mathrm{ for }\ (\sigma_1,\sigma_2) \in \mathcal{U},
\end{equation}
then there exists $c_0 > 0$ (depending only on $\lambda_1,\lambda_2,\kappa,\mathcal{U}$, and growth of $W$) such that for any $G \in \R^{3\times 2}$ 
\begin{equation}\label{eq86}
 \wr(G) - \wr(F) - D\wr(F):(G-F) \ge c_0g_p\left(|(G-F)\cdot n|\right).
\end{equation}
\end{lemma}

\begin{proof}
We first prove the statement for large strains (i.e. if $G$ is large). Let $\sigma_1 \ge \sigma_2$ are the singular values of $G$. We observe that the coercivity of $\w(G)$ (see~\eqref{eq032}) implies that $\wr(G)$ has also $p$-th power growth for large matrices. Indeed, if $G$ has only tensile stress, we have $\w(G) = \wr(G)$, i.e. $\wr(G)$ has the same growth as $\w$. Otherwise, we know from the ordered force inequality~\eqref{c2} that
\begin{equation}\nonumber
 f_r(\sigma_1,\sigma_2) \ge f_r(\sigma_1 - 1,1) \ge c (|G|^p - 1),
\end{equation}
where the last inequality follows from the coercivity of $\w$ and from the fact that strains $(\lambda_1-1,1)$ produce only tensile stresses (and so $\w = \wr$ in this case). Thus the LHS of~\eqref{eq86} grows at least like $|G|^p$. Since the RHS has at most such growth, the conclusion follows.

It remains to prove the statement in the case \[|G| \le M\] for some $M$. First, we observe that in this case $|(G-F)\cdot n|^2$ and $|(G-F)\cdot n|^p$ are comparable, and so we can replace $g_p$ in~\eqref{eq86} with a quadratic function, i.e. we need to show
\begin{equation}\label{eq86+}
 \wr(G) - \wr(F) - D\wr(F):(G-F) \ge c_0|(G-F)\cdot n|^2
\end{equation}
for any $|G| \le M$ (possibly with a different $c_0$ than in~\eqref{eq86}). We start by computing $D\wr(F)$. Since $\wr$ is rotationally invariant, we can assume without loss of generality that 
\begin{equation}\nonumber
 F = \begin{pmatrix} \lambda_1 & 0 \\ 0 & \lambda_2 \\ 0 & 0\end{pmatrix}, n = \begin{pmatrix} 1 \\ 0 \end{pmatrix}.
\end{equation}
Then a simple calculation reveals that
\begin{equation}\nonumber
 D\wr (F) = \begin{pmatrix} \alpha_1 & 0 \\ 0 & \alpha_2 \\ 0 & 0\end{pmatrix},
\end{equation}
where $\alpha_1 = \partial_1 f_r(\lambda_1,\lambda_2)$ and $\alpha_2 = \partial_2 f_r(\lambda_1, \lambda_2)$. We observe that $\alpha_1 > 0$, $\alpha_2 \ge 0$, and that $\alpha_2 = 0$ iff $\lambda_2 \le \o (\lambda_1)$.
We rewrite~\eqref{eq86+}:
\begin{align}\label{eq86a}
 \wr(G) - \wr(F) + D\wr(F):F &\ge D\wr(F):G + c_0|(G-F)\cdot n|^2 \\
 &= \alpha_1G_{11} + \alpha_2G_{22} + c_0\left( \left(G_{11} - \lambda_1\right)^2 + G_{21}^2 + G_{31}^2 \right) \nonumber \\
 &= \left(\alpha_1 - 2c_0\lambda_1\right) G_{11} + \alpha_2G_{22} + c_0\left( G_{11}^2 + G_{21}^2 + G_{31}^2 \right) + c_0\lambda_1^2 \nonumber
\end{align}

We choose $c_0 > 0$ small enough so that $\alpha_1 - 2c_0\lambda_1 > 0$. The LHS of the inequality depends only on singular values of $G$. 
Hence, we can prove~\eqref{eq86a} by maximizing the RHS among all matrices $G$ with given singular values $\sigma_1 \ge \sigma_2$.  
We give the argument assuming that $\alpha_1 - 2c_0\lambda_1 \ge \alpha_2$ (the proof in the case $\alpha_1 - 2c_0\lambda_1 \le \alpha_2$ is analogous).

To prove~\eqref{eq86a} we will use the following lemma due to von Neumann~(see, e.g, \cite{bib-mirsky}):

\begin{lemma}\label{lm_neumann}
 If $A,B$ are $n \times n$ matrices with singular values
\[ \sigma_1 \ge \dots \ge \sigma_n, \quad \rho_1 \ge \dots \ge \rho_n \]
respectively, then
\[ |\mathrm{tr}(AB)| \le \sum_{r=1}^n \sigma_r \rho_r. \]
\end{lemma}

We shall apply the lemma to find the maximal possible value of the RHS~\eqref{eq86a} among all matrices $G$ with singular values $\sigma_1 \ge \sigma_2$. First we set $A := (G|\mathbf{0})$ (the $3 \times 3$ matrix with first two columns identical with $G$ and the third column equal $0$) and $B = \mathrm{diag}(\alpha_1-2c_0\lambda_1,\alpha_2,0)$. Then the lemma gives
\begin{equation}\nonumber
 \left(\alpha_1-2c_0\lambda_1\right) G_{11} + \alpha_2 G_{22} \le \left(\alpha_1-2c_0\lambda_1\right) \sigma_1 + \alpha_2 \sigma_2.
\end{equation}
For $A := G^TG$ and $B = \mathrm{diag}(1,0)$ the lemma implies 
\begin{equation}\nonumber
 G_{11}^2 + G_{21}^2 + G_{31}^2 \le \sigma_1^2.  
\end{equation}
Together we see that the RHS of~\eqref{eq86a} is at most $(\alpha_1 - 2c_0\lambda_1)\sigma_1 + \alpha_2\sigma_2 + c_0\sigma_1^2 + c_0\lambda_1^2$. To see that this bound is optimal, we use
$G_0 := \begin{pmatrix} \sigma_1 & 0 \\ 0 & \sigma_2 \\ 0 & 0 \end{pmatrix}$ as $G$. We got that the RHS of~\eqref{eq86a} is maximal for the choice $G = G_0$, and so we need to prove~\eqref{eq86+} only for $G = G_0$ for any $\sigma_1 \ge \sigma_2 \ge 0$, i.e. to show
\begin{equation}\label{eqlm1}
 f_r\left(\sigma_1,\sigma_2\right) - f_r(\lambda_1,\lambda_2) - Df_r(\lambda_1,\lambda_2)(\sigma_1 - \lambda_1,\sigma_2 -\lambda_2) \ge c_0(\sigma_1 - \lambda_1)^2.
\end{equation}
Using Taylor's expansion of $f_r$ at the point $(\lambda_1,\lambda_2)$, and assumptions~\eqref{eq032} and~\eqref{eq85} we see that
\begin{align*}
 f_r\left(\sigma_1,\sigma_2\right) - f_r(\lambda_1,\lambda_2) - Df_r(\lambda_1,\lambda_2)(\sigma_1 - \lambda_1,\sigma_2 -\lambda_2) &\ge
 \int_{(\lambda_1,\lambda_2)}^{(\sigma_1,\sigma_2)} \left< D^2 f_r(\xi)\left( \sigma - \xi \right),\sigma - \xi \right> \mathrm{d}\xi  \\
 &\ge c_1\kappa (\sigma_1 - \lambda_1 )^2,
\end{align*}
where the last inequality follows from the fact that $D^2 f_r \ge \kappa e_1 \otimes e_1$ in a non-trivial part of the segment between $(\lambda_1,\lambda_2)$ and $(\sigma_1,\sigma_2)$ (remember that $G$ is bounded), and $D^2 f_r \ge 0$ otherwise. This completes the proof of the lemma since we showed that~\eqref{eqlm1} holds with $c_0 = c_1\kappa$.\footnotedva{new last sentence - below~\eqref{eqlm1}}
\end{proof}


\begin{lemma}\label{lm2}
Assume~\eqref{eq032} holds. 
Let $F \in \R^{3\times2}$ be an orthogonal matrix with singular values $1 < \lambda_2 \le \lambda_1 \le K$. 
If there exist $\kappa > 0$ and an open neighborhood $\mathcal{U}$ of $(\lambda_1,\lambda_2)$ such that 
\begin{equation}\nonumber
 D^2 f_r(\sigma_1,\sigma_2) \ge \kappa I \ \mathrm{ for }\ (\sigma_1,\sigma_2) \in \mathcal{U},
\end{equation}
then there exists $c_0 > 0$ (depending only on $\lambda_1,\lambda_2,\kappa,\mathcal{U},K$, and growth of $W$) such that for any $G \in \R^{3\times 2}$ 
\begin{equation}\label{eq106}
 \wr(G) - \wr(F) - D\wr(F):(G-F) \ge c_0g_p\left(|G-F|\right).
\end{equation}
\end{lemma}

\begin{proof}
The idea of the proof is simple and resembles proof of the previous lemma. If $G$ is large, we get the statement the same way as in the previous lemma. 

Otherwise we can assume $|G| \le M$ for some (possibly large) $M$.
For any such $G$ the part of the segment connecting $F$ and $G$ (i.e. $F + t(G-F)$ for $t\in (0,1)$) which belongs to $\mathcal{U}$ will be at least $\epsilon$ part of the whole segment. 
The LHS of~\eqref{eq106} can be written as
\begin{equation}\nonumber
 \int_0^1 \left< D^2\wr(F+t(G-F))(G-F),G-F\right>(1-t) \ud t, 
\end{equation}
and the integral is at least $\kappa(1-t) |G-F|^2$ along the non-trivial part of the segment and non-negative everywhere else. Therefore~\eqref{eq106} follows.
\end{proof}

Before proceeding to step 3, we need one more lemma, that is similar in character to the preceding ones but involves $\w$ instead of $\wr$. 

\begin{lemma}\label{lm3}
Assume~\eqref{eq032} holds. Let $F \in \R^{3\times2}$ be an orthogonal matrix with singular values $\lambda_1 > 1$, $\lambda_2 \le \o (\lambda_1)$, and $n$ be the right singular vector corresponding to $\lambda_1$. If there exists $\kappa > 0$ and an open neighborhood $\mathcal{U}$ of $(\lambda_1,\lambda_2)$ such that
\begin{equation}\nonumber
 D^2 f_r(\sigma_1,\sigma_2) \ge \kappa e_1 \otimes e_1 \ \mathrm{ for }\ (\sigma_1,\sigma_2) \in \mathcal{U},
\end{equation}
then there exists $c_0 > 0$ such that for any $G \in \R^{3\times 2}$ 
\begin{equation}\label{eq110}
 \w(G) - \w(F_0) - D\w(F_0):(G-F) \ge c_0g_p\left(\mathrm{dist}(G,SO(3)F_0)\right),
\end{equation}
where $F_0 n = F n$ and $\wr(F)=\w(F_0)$.
\end{lemma}

\begin{proof}
 Let $H \in \R^{3\times 2}$ be such that $H n = G n$ and $\w(H)$ is minimal among all such $H$. We observe
\begin{gather*}
 \w(H) \ge \wr(H), \quad  \w(F_0) = \wr(F_0). 
\end{gather*}
Arguing as we did in the proof of Lemma~\ref{lm1} (i.e. using rotational invariance of $\w$ and $\wr$ to compute $D\w$ and $D\wr$) we also have
\begin{equation}\nonumber
DW(F_0):(G-F) = D\wr(F_0):(H-F_0).
\end{equation}	

We claim that 
\begin{equation}\nonumber
 \w(G) - \w(H) \ge cg_p(\dist(G,SO(3)H))
\end{equation}
follows from~\eqref{eq032}. We give a proof of this fact in the Appendix (see Lemma~\ref{lm6} and Remark~\ref{rmk1}).

Now Lemma~\ref{lm1} and the previous inequality imply that the LHS of~\eqref{eq110} is at least
\begin{multline}\label{eq111}
 \left[\w(G) - \w(H)\right] + \left[\wr(H) - \wr(F_0) - D\wr(F_0):(H-F_0)\right] \ge \\
 c_0\left( g_p \left( \mathrm{dist}(G,SO(3)H\right) + g_p \left( |(H - F_0)\cdot n|\right) \right)
\end{multline}
($c_0$ is a generic constant, i.e. it can change from line to line).

We claim that in the second term on the RHS of~\eqref{eq111} we have
\begin{equation}\nonumber
 |(H-F_0)\cdot n| \ge c\dist(H,SO(3)F_0).
\end{equation}
To prove it, we first observe that without loss of generality we can assume $n=(1,0)$. Then using~\eqref{eq007} we see that 
\begin{gather*}
 H = (v_1 | v_2) \textrm{ and } F_0 = (w_1 | w_2),
\end{gather*}
where $v_1 \perp v_2$, $|v_2| = g(|v_1|)$ and $w_1 \perp w_2$, $|w_2| = g(|w_1|)$. There clearly exists a rotation $R \in SO(3)$ such that
$v_1 || Rw_1$ and $v_2 || Rw_2$, and the vectors have the same orientation; thus
\begin{multline*}
 \dist(H,SO(3)F_0)^2 \le |v_1 - Rw_1|^2 + |v_2 - Rw_2|^2 = (|v_1| - |w_1|)^2 + (|v_2| - |w_2|)^2 \\ = (|v_1| - |w_1|)^2 + \left(g(|v_1|) - g(|w_1|)\right)^2 \le c(|v_1| - |w_1|)^2 \le c|(H-F_0)\cdot n|^2.
\end{multline*}
Hence the RHS of~\eqref{eq111} satisfies
\begin{multline}\nonumber
c_0\left( g_p \left( \mathrm{dist}(G,SO(3)H\right) + g_p \left( |(H - F_0)\cdot n|\right) \right) \\
\ge c_0\left[ g_p \left( \mathrm{dist}(G,SO(3)H\right) + g_p \left( \mathrm{dist}(H,SO(3)F_0)\right) \right] \ge c_0' g_p(\mathrm{dist}(G,SO(3)F_0)),
\end{multline}
where we used inequality~\eqref{gp2}. This completes the proof of Lemma~\ref{lm3}.
\end{proof}

We continue the proof of the lower bound~\eqref{eq76}.

{\bf\noindent Step 3:}
It is clear that there exists $\rho > 1$ such that
the region where the smaller singular value of $\grad u_0$ is at least $\rho$ is non-empty. We will denote this region $\NR$
and its complement in $\Omega$ as $\RR$. Lemmas~\ref{lm1} and~\ref{lm2} applied to~\eqref{eq84} imply 
\begin{equation}\nonumber
 \int_{\NR} g_p\left(|\grad u_h(x') - \grad u_0(x')|\right) \dxy + \int_{\RR} g_p\left(|(\grad u_h(x') - \grad u_0(x')\cdot n(x')|\right) \dxy \le C\delta h,
\end{equation}
where $C > 0$ depends only on $u_0$, $\wr$, and choice of $\rho$. We now use the Poincar\'e inequality adjusted to our setting (Lemma~\ref{lmPoincare}). We obtain
\begin{equation}\nonumber
 \int_{\NR} g_p(|u_h(x') - c_h - u_0(x')|) \dxy \le \int_{\NR} g_p(|Du(x') - Du_0(x')|) \dxy \le C\delta h
\end{equation}
for some $c_h$. 

We would like to extend the previous estimate into the whole $\Omega$. We fix a direction $\theta$ and define (using radial coordinates)
\begin{equation}\nonumber
 f(t) := u_h(t,\theta) - u_0(t,\theta),\quad t \in (R_{in},R_{out}).
\end{equation}
Let us call $M$ the radius of the boundary between $\NR$ and $\RR$, and let
\begin{equation}\nonumber
 K := \int_{R_{in}}^{R_{out}} g_p(|f'(t)|) \ud t + \int_{M}^{R_{out}} g_p(|f(t) - c_h|) \ud t.
\end{equation}
Then by the Poincar\'e inequality applied to $f$ on $(R_{in},R_{out})$ we get
\begin{equation}\nonumber
 \int_{R_{in}}^{R_{out}} g_p(|f(t) - c_\theta|) \ud t \le C \int_{R_{in}}^{R_{out}} g_p(|f'(t)|) \ud t \le CK
\end{equation}
for some $c_\theta$. We also have
\begin{multline}\nonumber
 (R_{out} - M)g_p(|c_h - c_\theta|) = \int_{M}^{R_{out}} g_p(|c_h - c_\theta|) \ud t \\ 
\le C\left( \int_{M}^{R_{out}} g_p(|f(t) - c_h|) \ud t + \int_{M}^{R_{out}} g_p(|f(t) - c_\theta|) \ud t\right) \le CK
\end{multline}
and so 
\begin{equation}\label{eq202}
 \int_{R_{in}}^{R_{out}} g_p(|f(t) - c_h|) \ud t \le C\left( \int_{R_{in}}^{R_{out}} g_p(|f(t) - c_\theta|) \ud t + \int_{R_{in}}^{R_{out}} g_p(|c_\theta - c_h|) \ud t \right) \le CK.
\end{equation}
Finally, by integrating~\eqref{eq202} in $\theta$ we obtain
\begin{equation}\label{eq-99}
 \int_{\Omega} g_p(|u_h(x') - c_h - u_0(x')|) \dxy \le C\delta h. 
\end{equation}
{\bf\noindent Step 4:}
The next step in the proof is the interpolation between $||u_{h,3}||_{L^2(\Omega)}$ and $||\grad^2 u_{h,3}||_{L^2(\Omega)}$. For that reason we assume $p=2$ (instead of a more general $1 < p \le 2$) for the rest of this section. (Since the previous lemmas will be used later (see Section~\ref{sec3d}), we proved them assuming only $1 < p \le 2$.)

When $p=2$, \eqref{eq-99} reads
\begin{equation}\nonumber
 ||u_h(x') - c_h - u_0(x')||_{L^2(\Omega)}^2 \dxy \le C\delta h. 
\end{equation}
Without loss of generality we can assume $c_h=0$, since our problem is translation invariant; in particular we have
\begin{equation}\label{interp1}
  ||u_{h,3}||_{L^2(\Omega)}^2 \le C\delta h.
\end{equation}
Since 
\begin{equation}\nonumber
 E_h(u_h) - h^2||\grad^2 u_{h,3}||_{L^2(\Omega)}^2 \ge \min_u E_0(u) = \mathcal{E}_0, 
\end{equation}
we have that
\begin{equation}\label{interp2}
 ||\grad^2 u_{h,3}||_{L^2(\Omega)}^2 \le \delta/h.
\end{equation}
Interpolating between \eqref{interp1} and \eqref{interp2} we obtain
\begin{equation}\label{outofplane}
 ||\grad u_{h,3}||_{L^2(\Omega)}^2 \le C\delta.
\end{equation}

{\bf\noindent Step 5:}
We want to estimate
\begin{equation}\nonumber
 \int_\Omega \left| \w(\grad u_h) - \w(\grad u_h^{12})\right| \dxy,
\end{equation}
where $u_h^{12} = (u_{h,1},u_{h,2},0)$ is the projection of $u_h$ into $x$-$y$ plane  and $A(x') := \grad u_h(x') - \grad u_h^{12}(x')$. By~\eqref{w=f} the previous integral is equal to 
\begin{equation}\nonumber
 \int_\Omega \left| f(\sigma_1(x'),\sigma_2(x')) - f(\lambda_1(x'),\lambda_2(x'))\right| \dxy,
\end{equation}
where $\sigma_1(x'),\sigma_2(x')$ and $\lambda_1(x'),\lambda_2(x')$ are the singular values of $\grad u_h(x')$ and $\grad u_h^{12}(x')$, respectively. Since the singular values of a matrix are Lipschitz functions of the corresponding matrix (see, e.g., Corollary 8.6.2 in~\cite{golubMatrix}), we have
\begin{equation}\label{fdiff}
 |f(\sigma_1(x'),\sigma_2(x')) - f(\lambda_1(x'),\lambda_2(x'))| = |\grad f(\xi(x')) (\sigma_1(x')-\lambda_1(x'),\sigma_2(x')-\lambda_2(x'))| \le C|\grad f(\xi(x'))| |A(x')|,
\end{equation}
where $\xi(x')$ is a point on the segment connecting $(\sigma_1(x'),\sigma_2(x'))$ and $(\lambda_1(x'),\lambda_2(x'))$. By~\eqref{eq032} $\grad f(1,1) = 0$ and $D^2f \le C$, and so $Df (\zeta) \le C(|\zeta| + 1)$. Using quadratic growth of $f$ (\eqref{eq032} with $p=2$) we obtain
\begin{equation}\label{Df}
 |Df(\xi(x'))|^2 \le C(|\xi(x')|^2 + 1) \le C'(f(\xi(x')) + 1) \le C'(f(\sigma_1(x'),\sigma_2(x')) + f(\lambda_1(x'),\lambda_2(x')) + 1),
\end{equation}
where the last inequality trivially follows from the convexity of $f$. Integrating~\eqref{fdiff} and using~\eqref{Df} together with H\"older inequality we get
\begin{multline}\nonumber
 \int_\Omega \left| f(\sigma_1(x'),\sigma_2(x')) - f(\lambda_1(x'),\lambda_2(x'))\right| \dxy \le C\left( \int_\Omega |Df(\xi(x'))|^2 \dxy \right)^{1/2} \left( \int_\Omega |A(x')|^2 \dxy \right)^{1/2} \\ \le C\left( \int_\Omega W(Du_h) + W(Du_h^{12}) \dxy + 1\right)^{1/2} ||Du_{h,3}||_{L^2(\Omega)}.
\end{multline}
Therefore by~\eqref{outofplane} and using $\delta \le 1, h \le 1$ we have that
\begin{multline}\nonumber
  \left( \int_\Omega \left| \w(\grad u_h) - \w(\grad u_h^{12})\right| \dxy \right)^2 \le C \left( 2 \int_\Omega W(Du_h) \dxy +  \int_\Omega \left| \w(\grad u_h) - \w(\grad u_h^{12})\right| \dxy + 1 \right) \delta \\ 
  \le C \delta \left( \mathcal{E}_0 + \delta h + 1 + \int_\Omega \left| \w(\grad u_h) - \w(\grad u_h^{12})\right| \dxy \right) \le 
  C'\delta\left( 1 + \int_\Omega \left| \w(\grad u_h) - \w(\grad u_h^{12})\right| \dxy  \right).
\end{multline}
It follows easily that
\begin{equation}\nonumber
  \int_\Omega \left| \w(\grad u_h) - \w(\grad u_h^{12})\right| \dxy \le C \delta^{1/2},
\end{equation}
and so
\begin{equation}\label{eq004}
\begin{aligned}
 \left| E_h(u_h) - E_h(u_h^{12}) \right| &\le C \delta^{1/2} + h^2||\grad^2 u_{h,3}||_{L^2(\Omega)}^2 \le 
  C(\delta^{1/2} + \delta h) \le C\delta^{1/2}.
\end{aligned}
\end{equation}
Equations~\eqref{eq005} and~\eqref{eq004} together imply
\begin{equation}\nonumber
 E(u_h^{12}) - E_0(u_0) \le C\delta^{1/2},
\end{equation}
where $E(v) = \int_\Omega \w(\grad v) + B(v)$ ($v$ is just an in-plane deformation, so there is no bending term present). \\

{\bf\noindent Step 6:}
Observe that the last relation does not depend on thickness $h$ anymore. In fact, to finish the proof we just need to show that the minimum of the energy $E_h$ over in-plane deformations has to be strictly larger than the minimum of the relaxed energy $E_0$:

\begin{lemma}\label{lm9}
 Under the assumptions of Theorem~\ref{thm-2d} we have
\begin{equation}\nonumber
 \min_{u \in W^{2,2}(\Omega,\R^2\times \{ 0 \} )} E_h(u) > \min_{u \in W^{2,2}(\Omega,\R^3)} E_0(u) = \mathcal{E}_0.
\end{equation}
\end{lemma}

\begin{proof}
Let us assume the contrary, i.e. for any small $\delta > 0$ there exists a function $u : \Omega \to \R^2$ such that
\begin{equation}\nonumber
 \int_\Omega \w (\grad u) \dxy + B(u) \le \mathcal{E}_0 + \delta = \int_\Omega \wr (\grad u_0) \dxy + B(u_0) + \delta.
\end{equation}

The plan is to obtain a contradiction by showing that the areas of $u(\relaxed)$ and $u_0(\relaxed)$ should be very similar using one argument and at the same time very different for another reason (here $\relaxed$ is the relaxed region introduced near~\eqref{eq-relaxed}).
First, using Euler-Lagrange equation for $u_0$ we can replace the boundary term $B(u - u_0)$ by the gradient term:
\begin{equation}
 \int_\Omega \w(\grad u) - \wr(\grad u_0) -D\wr(\grad u_0)\cdot (\grad u - \grad u_0) \dxy \le \delta. \nonumber
\end{equation}
Since $\w \ge \wr$ and $\wr$ is convex, the integrand in the last relation is non-negative a.e. Therefore the last relation remains true if we integrate over the relaxed region $\relaxed$ instead of the whole domain $\Omega$. To proceed, we would like to find a matrix $F_0$ such that $\wr(\grad u_0) = \w(F_0)$ (i.e. we want to relax compressive stresses in $\grad u_0$ if they are present). We know that $\grad u_0$ has compressive stresses in the hoop direction (and tensile in the radial direction), which means that $F_0$ and $\grad u_0$ coincide in the radial direction and are different in the hoop direction, i.e. 
\[F_0(x') n(x') = (\grad u_0(x')) n(x')\quad \mathrm{ and }\quad F_0(x') n^\perp(x') = c n^\perp(x')\]
with some $c > (Du_0 n^\perp)n^\perp$. Moreover, it can be easily seen that in this case $D\wr(\grad u_0) = D\w(F_0)$. We rewrite the previous inequality to obtain
\begin{equation}
 \int_\relaxed \w(\grad u(x')) - \w(F_0(x')) -D\w(F_0(x'))\cdot (\grad u(x') - \grad u_0(x')) \dxy \le C\delta. \nonumber
\end{equation}
Using Lemma~\ref{lm3} with $p=2$ we see that 
\begin{equation}\nonumber
 \int_\relaxed \dist^2(\grad u,SO(3)F_0) \dxy \le C\delta.
\end{equation}
Whence
\begin{equation}\label{eq-c}
\left| |u(\relaxed)| - \int_\relaxed \det F_0 \dxy \right| \le \int_\relaxed \left| \det \grad u - \det F_0 \right| \dxy \le \ 
C \left(\delta + \delta^{1/2}\right) \le C\delta^{1/2}.
\end{equation}
Further we see that 
\begin{equation}\label{eq-b}
 \det F_0 > \det \grad u_0.
\end{equation}
Indeed, $F_0$ and $\grad u_0$ are ``same'' in the radial direction (i.e. $F_0 n = (\grad u_0) n$), whereas due to compression in the hoop direction in the relaxed solution $u_0$ we have ($F_0 n^\perp)n^\perp > ((\grad u_0)n^\perp)n^\perp$. Hence we get~\eqref{eq-b} by taking product of the two previous relations. Finally integrating~\eqref{eq-b} we obtain
\begin{equation}\label{eq-d}
 \int_\relaxed \det F_0 \dxy  > \int_\relaxed \det \grad u_0 \dxy = |u_0(\relaxed)|.
\end{equation}

To finish the proof we want to show 
\begin{equation}\label{eq-e}
 \left| u(\relaxed) \right| \le \left|u_0(\relaxed)\right| + C\delta^{1/3}.
\end{equation}
Then by combining~\eqref{eq-c} and~\eqref{eq-e} we get
\begin{equation}\nonumber
 \int_\relaxed \det F_0 \dxy  < |u_0(\relaxed)| + C\delta^{1/3},
\end{equation}
contradicting~\eqref{eq-d} since $\delta > 0$ can be arbitrary small.

To show~\eqref{eq-e} we set $\epsilon := \delta^{1/3}$ and define
\begin{equation}\nonumber
 \mathcal{M} := u(\relaxed) \cap \left\{ x \in \R^2 : \dist (x,u_0(\relaxed)) > \epsilon \right\}.
\end{equation}
Then
\begin{equation}\label{eq54}
 |u(\relaxed)| \le \left| \left\{ x \in \R^2 : \dist (x,u_0(\relaxed)) \le \epsilon \right\} \right| + \left|\mathcal{M}\right| \le |u_0(\relaxed)| + C\epsilon + \left|\mathcal{M}\right|.
\end{equation}
It is enough to estimate the size of $\mathcal{M}$. Arguing as before (c.f.~\eqref{eq-99}), we know that
\begin{equation}\nonumber
 ||u-u_0||_{L^2(\relaxed,\R^2)}^2 \le C\delta,
\end{equation}
whence 
\begin{equation}\label{eq53}
 \epsilon^2 \left|u^{-1}(\mathcal{M})\right| \le \int_{u^{-1}(\mathcal{M})} |u-u_0|^2 \dxy \le ||u-u_0||_{L^2(\relaxed,\R^2)}^2 \le C\delta.
\end{equation}
By virtue of~\eqref{eq53} 
\begin{align}\nonumber
 |\mathcal{M}| &= \int_{u^{-1}(\mathcal{M})} \det \grad u \dxy = \int_{u^{-1}(\mathcal{M})} \det F_0 \dxy + \int_{u^{-1}(\mathcal{M})} \det \grad u - \det F_0 \dxy \\ &\le C\left|u^{-1}(\mathcal{M})\right| + \int_{\relaxed} \left| \det \grad u - \det F_0 \right| \dxy \le C\delta \epsilon^{-2} + C\delta^{1/2} \le C\delta^{1/3}.\nonumber
\end{align}
Using this estimate in~\eqref{eq54} we obtain~\eqref{eq-d}.
\end{proof}

This completes the proof of the matching lower and upper bound in the two-dimensional setting~\eqref{eq-2d}.

\section{The 3D result}\label{sec3d}

In this section we will prove the scaling law for the minimum of the elastic energy in the nonlinear three-dimensional setting.
As in the previous section, we need to show an upper and a lower bound. As usual in problems of this type, the upper bound is an easy consequence of the upper bound for the $2D$ setting. The main goal is therefore to show the lower bound in this more general setting. 

As explained in Section~\ref{sec-model} we consider a nonlinear $3D$ energy
\begin{equation}\nonumber
 \eg(u) := \frac{1}{h}\int_{\Omega \times (0,h)} \w_{3D}(\grad u) \dxyz
\end{equation}
instead of the reduced $2D$ energy~\eqref{eq001}. The boundary term in the $3D$ setting is defined as
\begin{equation}\nonumber
 \B(u) := \frac{T_{in}}{h} \int_{|\hat x|=R_{in}} u(x) \cdot\frac{\hat x}{R_{in}} \dS - \frac{T_{out}}{h} \int_{|\hat x|=R_{out}} u(x) \cdot\frac{\hat x}{R_{out}} \dS.
\end{equation}
The main result of this section is
\begin{theorem}\label{thm-3d}
Under the hypothesis of Theorem~\ref{thm-rel} there exist constants $0 < C_1 < C_2$ independent of $h$ such that
\begin{equation}\nonumber
 \mathcal{E}_0 + C_1h \le \min_{u \in W^{1,p}(\Omegah)} \eg(u) + \B(u) \le \mathcal{E}_0 + C_2h.  
\end{equation}
\end{theorem}

\subsection{The upper bound in the three-dimensional setting}

The construction of a test function $\ug$ in the three-dimensional setting is based on the test function $u_h$ defined in the previous section to show the upper bound in the two-dimensional setting. Following proposed Kirchhoff-Love ansatz, the normal to the mid-surface remains straight and normal to the mid-surface after deformation. Therefore, we just need to find how much should each of these normals stretch. It follows from the definition of $\w$ (see~\eqref{eq002}) that for any $x' \in \Omega$ there exists a factor $\alpha(x')$ such that the vector $\alpha(x')\nu(x')$ satisfies
\begin{equation}\nonumber
 \w(\grad u_h(x')) = \wt(\grad u_h(x')|\alpha(x')\nu(x')),
\end{equation}
where $\nu(x')$ is the unit normal to $u_h(\Omega)$ at $u_h(x')$. We observe that $\alpha(x')$ is bounded and $|\grad \alpha| \le C|\grad^2 u_h|$. 
We define the solution $\ug$ as
\begin{equation}\nonumber
 \ug(x) := u_h(x') + x_3\cdot\alpha(x')\nu(x')
\end{equation}
and compute
\begin{equation}\nonumber
 \grad \ug(x) = \left( \grad u_h | \alpha \nu \right) + x_3 \cdot (\grad (\alpha \nu) | 0).
\end{equation}
Then
\begin{align}\nonumber
 \eg(\ug) - E_h(u_h) &\le \frac{1}{h} \int_\Omegah \wt(\ (\grad u_h|\alpha\nu) + x_3\cdot(\grad (\alpha\nu)|0)\ ) - \wt(\grad u_h|\alpha\nu) \dxyz \\ \nonumber
 &\le \frac{1}{h} \int_\Omegah x_3 \cdot D\wt(\grad u_h|\alpha\nu):(\grad (\alpha\nu)|0) + x_3^2\cdot C|\grad (\alpha\nu)|^2 \dxyz \\ \nonumber
 &\le Ch^2 \int_\Omega |\grad (\alpha\nu)|^2 \dxy \le Ch^2 \int_\Omega \alpha(x')^2|\grad \nu|^2 + |\grad \alpha|^2 \dxy
\end{align}
where $C$ depends on $||D^2\wt||_{L^\infty}$. We know that $h^2 \int_\Omega |\grad^2 u_{h}|^2$ is bounded by $Ch$ (see Remark~\ref{rmk2}), and thus using boundedness of $\grad u_h$ and 
\begin{gather}\nonumber
 \int_\Omega |\grad \nu|^2 \dxy \le C\int_\Omega |\grad^2 u_h|^2 \dxy 
\end{gather}
we obtain $h^2\int_\Omega \alpha^2|\grad \nu|^2 \dxy \le Ch$. A similar estimate is true for the second term:
\begin{equation}\nonumber
 h^2 \int_\Omega |\grad \alpha|^2 \dxy \le Ch^2 \int_\Omega |\grad^2 u_h|^2 \dxy \le Ch.
\end{equation}
Together we have obtained 
\begin{equation}\nonumber
 \eg(\ug) \le E_h(u_h) + Ch \le \mathcal{E}_0 + Ch.
\end{equation}

\subsection{The lower bound in the three-dimensional setting}\label{lb3D}

Our goal is to show a lower bound similar to~\eqref{eq76}:
\begin{equation}\label{eq129}
\min_{u \in W^{1,p}(\Omegah)} \eg(u) + \B(u) \ge \mathcal{E}_0 + ch.
\end{equation}
To first approximation, the proof of~\eqref{eq129} consists of slicing our domain, applying the two-dimensional lower bound on each slice, then patching them together. But two new features require changes in the argument. First, the energy density $\w$ in $2D$ was derived from the $3D$ energy density $\wt$ assuming the missing third component is optimal. Therefore we need to estimate how much the third component of $\grad u$ differs from the optimal one, and how the optimal vector 
depends on the first two components. As a second feature, where $2D$ used interpolation we need to proceed differently by using a rigidity theorem. 

From now on let $h > 0$ be fixed and consider a function $\ut \in W^{1,p}(\Omegah)$. Then for any fixed $x_3 \in (-h/2,h/2)$ we define
\begin{align}\nonumber
u(x') &:= \ut(x',x_3)  
\end{align}
for $x' \in \Omega$. 
Since $u_0$ is the minimizer of the relaxed energy $E_0$ we have that
\begin{equation}\nonumber
 \mathcal{E}_0 = E_0(u_0) \le E_0(u) \le \int_\Omega \w(\grad u) \dxy + B(u)
\end{equation}
and we set
\begin{equation}\nonumber
 R := \int_\Omega \w(\grad u) \dxy + B(u) - \mathcal{E}_0 \ge 0.
\end{equation}

Our initial goal is to show that if $R$ is small then $\grad u$ is close to some function (which is derived from $\grad u_0$).
We use Euler-Lagrange equation to replace boundary term $B$ by $D\wr$ and obtain
\begin{equation}\nonumber
 \int_\Omega \w(\grad u) - \wr(\grad u_0) - D\wr(\grad u_0) (\grad u - \grad u_0) \dxy = R.
\end{equation}
Using notation of Lemma~\ref{lm3} this can be rewritten as
\begin{equation}\nonumber
 \int_\Omega \w(\grad u) - \w(F_0) - D\w(F_0) (\grad u - \grad u_0) \dxy = R,
\end{equation}
where $F_0(x) \cdot n(x) = \grad u_0(x) \cdot n(x)$ and $\w(F_0(x)) = \wr(\grad u_0(x)).$
We apply Lemma~\ref{lm3} to get
\begin{equation}\label{eq065}
 \int_\Omega g_p\left(\dist(\grad u(x'),SO(3)F_0(x'))\right) \dxy \le CR.
\end{equation}
Now we go back to the $3D$ body. We see that
\begin{multline}\nonumber
 \eg(\ut) + \B(\ut) - \mathcal{E}_0 = \frac{1}{h}\int_{-h/2}^{h/2} \int_\Omega \wt(\grad \ut(x)) - \w(\grad \ut(x) P) \dxyz + \\ \frac{1}{h} \int_{-h/2}^{h/2} \left(\int_\Omega \w(\underbrace{\grad \ut(x)P}_{Du(x')}) \dxy + B(\ut(\cdot,x_3)) - \mathcal{E}_0 \right) \dz =: I_1 + I_2,  
\end{multline}
where $P = \begin{pmatrix} 1 & 0 \\ 0 & 1 \\ 0 & 0 \end{pmatrix}$. Let $\zeta(x) \in \R^2$ be such that
\begin{equation}\nonumber
 \wt (\grad \ut (x)P|\zeta (x)) = \w (\grad \ut(x) P).
\end{equation}
Then Lemma~\ref{lm7} implies
\begin{equation}\label{eq066}
 I_1 \ge \frac{C}{h} \int_{-h/2}^{h/2} \int_\Omega g_p(|\grad \ut (x) - (\grad \ut(x)P|\zeta(x))|) \dxyz,
\end{equation}
and by virtue of~\eqref{eq065}
\begin{equation}\nonumber
 I_2 \ge \frac{C}{h} \int_{-h/2}^{h/2} \int_\Omega g_p(\dist(\grad \ut(x)P, SO(3)F_0(x'))) \dxyz.
\end{equation}
We want to extend $F_0(x') \in \R^{3\times2}$ into a $3\times 3$ matrix. To do that, we find a vector $V \in \R^3$ such that $\det \left(F_0(x')|V\right) > 0$ and
\begin{equation}\nonumber
 \wt(F_0(x')|V) = \min_{\xi \in \R^3} \wt(F_0(x')|\xi).
\end{equation}
Observe that the choice of $V$ is unique. We define 
\begin{equation}\label{eq-M}
 M(x') := (F_0(x')|V).
\end{equation}
Lemma~\ref{lm8} then implies
\begin{equation}\nonumber
 \dist ( (\grad \ut(x)P|\zeta(x)), SO(3)M(x')) \le C \dist(\grad \ut(x)P, SO(3)F_0(x')),
\end{equation}
and so
\begin{equation}\label{eq067}
I_2 \ge \frac{C}{h} \int_{-h/2}^{h/2} \int_\Omega g_p(\dist ( (\grad \ut(x)P|\zeta(x)), SO(3)M(x')),) \dxyz.
\end{equation}
By adding~\eqref{eq066} and~\eqref{eq067}, and by using~\eqref{gp2} (the triangle inequality for $g_p$ with factor $2$) we obtain
\begin{equation}\nonumber
 \eg(\ut) + \B(\ut) - \mathcal{E}_0 \ge \frac{C}{h}\int_{\Omegah} g_p\left( \dist(\grad \ut(x),SO(3)M(x'))\right) \dxyz.
\end{equation}
Moreover, from the previous analysis (see~\eqref{eq-99}) we know that there exists a function $\tau$ of $x_3$ alone such that
\begin{equation}\nonumber
 \frac{1}{h}\int_{-h/2}^{h/2} \int_\Omega g_p\left(|\ut(x)-\tau(x_3)-u_0(x')|\right) \dxy \dz \le C(\eg(\ut) + \B(\ut) - \mathcal{E}_0).
\end{equation}
By adding those two inequalities we obtain
\begin{equation}\nonumber
 J(\Omega,\ut) \le C(\eg(\ut) + \B(\ut) - \mathcal{E}_0),
\end{equation}
where for $U \subset \Omega$ we define
\begin{multline}\nonumber
 J(U,\ut) := \\ \frac{1}{h}\int_{U\times (-h/2,h/2)} \!\! g_p\left( \dist (\grad \ut(x),SO(3)M(x'))\right) \dxyz + \frac{1}{h}\int_{U\times(-h/2,h/2)}\!\! g_p\left( |\ut(x) - \tau(x_3) - u_0(x')|\right) \dxyz.
\end{multline}

We will obtain the lower bound from the following important lemma:

\begin{lemma}\label{lm-lb}
 Let $h>0$ be sufficiently small and let $r_0$ be such that $(r_0,r_0+2h) \subset (R_{in},L)$. Then
\begin{equation}\nonumber
 J(\A(r_0,r_0+h),\ut) \ge \eta(r_0) h^2,
\end{equation}
where $\A(\alpha,\beta)$ is an annulus with radii $\alpha < \beta$ and $\eta=\eta(r_0) > 0$ is a decreasing
function of $r_0$. 
\end{lemma}

The proof of the lemma will be given in Section~\ref{sec-lemma}. The desired lower bound is an easy consequence. 
%
Indeed, for any $R_{in} < r_0 < L-2h$ we know by Lemma~\ref{lm-lb} that
\begin{equation}\nonumber
 J(\A(r_0,r_0+h),\ut) \ge \eta(r_0)h^2.
\end{equation}
Adding such inequalities for $r_0 = R_{in} + kh$ such that $R_{in} \le r_0 \le (R_{in}+L)/2-2h$ we obtain
that 
\begin{equation}\nonumber
 J(\A(R_{in},L),\ut) \ge \sum_{k=0}^K \eta(r_0(k))h^2 \ge \eta\left((R_{in}+L)/2\right)\sum_{k=0}^K h^2 \ge Ch
\end{equation}
for some $C>0$, where $K = (R_{in}-L)/2h - 2$ and we used monotonicity of $\eta$.

Besides proving Lemma~\ref{lm-lb} (see the next section), the proof of the lower bound~\eqref{eq129} is finished.

\subsection{Proof of Lemma~\ref{lm-lb}}\label{sec-lemma}

\begin{proof}
Let us first sketch the idea of the proof. We assume $\ut$ has small energy (i.e. $J(\ut)$ is small), and want to compare it with $u_0$. We take a collection of $h^{-1/2}$ neighboring cubes in the hoop direction, each cube with side $h$. Using a rigidity theorem we show that $\ut$ is almost constant on each cube and doesn't change much between the cubes. 

Since $\ut$ has small energy, we prove that often $D\ut$ is larger than $Du_0$ in the hoop direction. After integration in the hoop direction we obtain that $\ut - u_0$ can not be small in most of the cubes. 
On the other hand, $\ut$ having small energy implies that $\ut - u_0$ has to be small in the $L^2$ sense, contradicting the previous fact. 

For better understanding we split our proof into several steps.

{\bf\noindent Step 1:}
Let us consider a part of $\A(r_0,r_0+h)$ with the length approximately $2h^{1/2}$. We set \[k := [h^{-1/2}]\] and for a given $\theta \in (0,2\pi)$ we define $2k$ cubes in the radial coordinates
\begin{equation}\nonumber
 Q_i = (r_0,r_0+h) \times (\theta+i\sigma,\theta+(i+1)\sigma)\times(-h/2,h/2),
\end{equation}
where $i=1,\dots,2k$ and $\sigma = h/r_0$. We denote by $\mathcal{Q}$ the union of those cubes and set 
\begin{multline}\nonumber
 \mathcal{J} := h\,J(\mathcal{Q},\ut) + |\mathcal{Q}|h^2 \\ = \int_{\mathcal{Q}} g_p\left(\dist(\grad \ut(x),SO(3)M(x'))\right) \dxyz + \int_{\mathcal{Q}} g_p\left(|\ut(x) - \tau(x_3) - u_0(x')|\right) \dxyz + |\mathcal{Q}|h^2.
\end{multline}

{\bf\noindent Step 2:}
We claim that $\mathcal{J} \ge Ch^{3.5}$ for some positive $C$. To prove the claim, we shall suppose  that
\begin{equation}\label{eqj}
 \mathcal{J} = \epsilon h^{3.5},
\end{equation}
and give a lower bound for $\epsilon$.
Since $M$ is defined in terms of $u_0$ (see~\eqref{eq-M}), we have $|\grad M(x')|\le C$ 
and consequently
\begin{equation}\nonumber
 \int_{Q_i} g_p \left( \dist (\grad \ut(x),SO(3)M_i)\right) \dxyz \le C\int_{Q_i} g_p\left(\dist(\grad \ut(x),SO(3)M(x'))\right) \dxyz + C|Q_i|h^2,
\end{equation}
where $M_i := M(x_i)$ with $x_i$ being a point in $Q_i$ (e.g. center of $Q_i$). Using the rigidity estimates of~\cite{bib-ago-dalmaso-desimone} (see also~\cite{bib-rigid1}) we obtain a rotation $R_i$ on each cube $Q_i$ such that
\begin{equation}\label{eq149}
 \int_{\mathcal{Q}} g_p\left(|\grad \ut(x) - R_iM_i|\right) \dxyz \le C\mathcal{J}
\end{equation}
and
\begin{equation}\nonumber
 \sum_{i=1}^{2k-1} g_p\left( |R_{i}-R_{i+1}| \right) \le C\mathcal{J}h^{-3}.
\end{equation}
Using convexity of $g_p$ and Jensen's inequality we get
\begin{equation}\label{eq200}
 g_p\left( \frac{|R_\alpha - R_{\alpha+\beta+1}|}{\beta} \right) \le \frac{\sum_{i=\alpha}^{\alpha+\beta} g_p\left(|R_i - R_{i+1}|\right)}{\beta} \le C\mathcal{J}h^{-3}\beta^{-1} \le Ch^{0.5}.
\end{equation}
Since the RHS in the previous relation is small, so is the argument of $g_p$ on the LHS (in particular, it is smaller than 1), and so we can take square root of both sides to obtain
\begin{equation}\nonumber
 \left|R_\alpha - R_{\alpha+\beta+1}\right| \le \sqrt{C\mathcal{J}h^{-3}k} =: \delta.
\end{equation}
We also know
\begin{equation}\nonumber
 R_i \in SO(3), \quad |M_i|\le C, \quad |n| = 1, \quad |M_i - M_j| \le Ckh, \quad |n(x) - n(y)| \le Ckh.
\end{equation}
Therefore we can choose one rotation $R_*$ (e.g. among $R_i$), a matrix $M_*$ (among $M_i$) and a unit vector $n_*^\perp$ such that for any $x'$
\begin{equation}\label{eq150}
 \left| R_iM_i n(x')^\perp - R_*M_*n_*^\perp \right| \le C\left(\delta + kh\right).
\end{equation}

{\bf\noindent Step 3:}
For $j\in \left\{ 1,\dots,k\right\}$ we have
\begin{multline}\label{eq151}
 \left(\int_{Q_{j+k}} \ut(x) -\tau(x_3) - u_0(x') \dxyz\right)  - \left(\int_{Q_{j}} \ut(x) - \tau(x_3) - u_0(x') \dxyz\right) \\
 = \int_{D_j} \left(\grad \ut(x) - \grad u_0(x')\right)n(x')^\perp \varphi_j(x') \dxyz \\
 = \int_{D_j} \left[ \left(\grad \ut(x) - R_iM_i\right)n(x')^\perp\! +\! \left( R_iM_in(x')^\perp - R_*M_*n_*^\perp\right)\right. + \\ \left. \left(R_*M_*n_*^\perp - \grad u_0(x')n(x')^\perp \right) \right]\varphi_j(x') \dxyz,
\end{multline}
where $D_j := \bigcup_{i=j}^{j+k} Q_i$ and $0 \le \varphi_j(x') \le h$ is a weight coming from the integration (more precisely, $\varphi_j$ is linear in $Q_j$ going from $0$ to $h$, $\varphi_j(x') = h$ on $Q_{j+1},\dots,Q_{j+k-1}$, and then decays linearly from $h$ to $0$ in $Q_{j+N}$). We point out that $n^\perp$ used in~\eqref{eq151} means a unit vector in the orthoradial (hoop) direction. 
The first two parts of the last integral can be directly estimated from~\eqref{eq149} and~\eqref{eq150} using convexity of $g_p$ and Jensen's inequality as in~\eqref{eq200}:
\begin{equation}\label{eq153a}
\left| \int_{D_j} \left(\grad \ut(x) - R_iM_i\right)n(x')^\perp \varphi_j(x') \dxyz \right| \le Ch \mathcal{J}^{1/2}|D_j|^{1/2}
\end{equation}
and
\begin{equation}\label{eq153b}
\left| \int_{D_j} \left( R_iM_in(x')^\perp - R_*M_*n_*^\perp\right) \varphi_j(x') \dxyz \right| \le C\left|D_j\right|\left(\delta h + kh^2\right).
\end{equation}

{\bf\noindent Step 4:}
Now we show a lower bound for the remaining term $\int_{D_j}\!\! \left(R_*M_*n_*^\perp\! -\!\grad u_0(x)n(x')^\perp\!\right) \varphi_j(x') \dxyz$
using the fact that $\grad u_0$ is ``smaller'' in the hoop ($n^\perp$) direction than $R_*M_*$. 
Arguing as in the proof of~\eqref{eq-b} we see that for any unit planar vector $\chi$ we have
\begin{equation}\nonumber
 |M_*\chi| - |\grad u_0(x') n(x')^\perp| \ge \kappa > 0,
\end{equation}
where $\kappa$ depends monotonically only on the radial position $r = |x'|$ (and approaches $0$ as $r_0 \to L$).  Therefore (using $\kappa_0 := \kappa(r_0+h)$)
\begin{equation}\nonumber
\begin{aligned}
 \Bigg|\int_{D_j} \left(R_*M_*n_*^\perp -\grad u_0(x)n(x')^\perp\right) \Bigg. & \Bigg. \varphi_j(x') \dxyz\Bigg| 
\\ &\ge \left| R_*M_*n_*^\perp \int_{D_j} \varphi_j(x') \dxyz\right| - \left|\int_{D_j} Du_0(x') n(x')^\perp \varphi_j(x') \dxyz\right|\\
&\ge |M_*n_*^\perp|\int_{D_j} \varphi_j(x') \dxyz - \int_{D_j} \left|\grad u_0(x') n(x')^\perp\right|\varphi_j(x') \dxyz\\ 
&= \int_{D_j} \left(|M_*n_*^\perp| - \left|\grad u_0(x') n(x')^\perp\right|\right) \varphi_j(x') \dxyz \\ 
&\ge \int_{D_j} \kappa_0 \varphi_j(x') \dxyz \ge \kappa_0 h |D_j|/2.
\end{aligned}
\end{equation}

{\bf\noindent Step 5:}
Using~\eqref{eq153a} and~\eqref{eq153b} together with the last relation we see from~\eqref{eq151} that
\begin{multline}\label{eq154}
 \left(\int_{Q_{j+k}} \ut(x) -\tau(x_3) - u_0(x') \dxyz\right)  - \left(\int_{Q_{j}} \ut(x) - \tau(x_3) - u_0(x') \dxyz\right) \ge \\ \left|D_j\right| h \left( \kappa_0/2 - \delta - kh\right) - Ch \mathcal{J}^{1/2}|D_j|^{1/2}.
\end{multline}
From~\eqref{eqj} and $ \delta = \sqrt{C \mathcal{J} k h^{-3}} \approx \sqrt{\mathcal{J} h^{-3.5}} \le \sqrt{\epsilon}$, we see that $\kappa_0/2 - \delta - kh \ge \kappa_0/4 > 0$ for $\epsilon \lesssim C\kappa_0^2$ and small $h$.

To finish the argument we sum~\eqref{eq154} over $j = 1,\dots,k$ to obtain
\begin{equation}\label{eq201}
 Ck|\mathcal{Q}|h(\kappa_0/2 - \delta - kh) \le \int_\mathcal{Q} |\ut(x) - \tau(x_3) - u_0(x)| \dxyz + Chk\mathcal{J}^{1/2}|\mathcal{Q}|^{1/2},
\end{equation}
where we have used that $|\mathcal{Q}|$ and $|D_j|$ are comparable. 
Using the convexity of $g_p$ and Jensen's inequality we have
\begin{equation}\nonumber
 \int_\mathcal{Q} |\ut(x) - \tau(x_3) - u_0(x)| \dxyz \le \mathcal{J}^{1/2}|\mathcal{Q}|^{1/2},
\end{equation}
and so after plugging the values of $\delta$, $|\mathcal{Q}|$, and $k$ into~\eqref{eq201} we see that
\begin{equation}\nonumber
 h^3 \kappa_0 \le C\left(h^{1.25} + h^{1.75}\right)\mathcal{J}^{1/2} \le Ch^{1.25} \mathcal{J}^{1/2}.
\end{equation}
Hence we obtain 
\begin{equation}\nonumber
 \mathcal{J} \ge C\kappa_0^2 h^{3.5}\qquad\mathrm{and}\qquad J(\mathcal{Q},\ut) \ge C\kappa_0^2h^{2.5}.
\end{equation}
Finally we cover the annulus $\mathcal{A}(r_0,r_0+h)$ with approximately $2\pi r_0 h^{-1/2}$ distinct copies of $\mathcal{Q}$ to obtain
\begin{equation}\nonumber
  J(\A(r_0,r_0+h),\ut) \ge \eta(r_0)h^2,
\end{equation}
where $\eta(r_0) = C r_0 \kappa_0^2$.
\end{proof}

\section{Discussion}

We have identified the scaling law 
for the minimum of the energy. We have achieved this by (i) constructing a family of functions with low energy; (ii) proving a lower bound on the energy with the same scaling as in (i). 

In the construction for the upper bound we introduced a cascade of wrinkles to obtain the linear scaling in thickness. As we approach the boundary between relaxed and non-relaxed region, the amplitude of the out-of-plane deformation vanishes, and so does the bending energy. The decay of the amplitude balances the increase in the bending energy due to the increasingly fine scale wrinkles.

The ``naive construction'' discussed at the beginning of Section~\ref{sec-2d} has no cascade, but its energy is $\mathcal{E}_0 + Ch|\log h|$ rather than $\mathcal{E}_0 + Ch$. It is natural to ask whether something similar to our cascade is {\it required} to get the optimal scaling. 

This paper has focused on the energy scaling law. It is natural to ask how the energy is distributed more locally; for example, is the optimal distribution (with respect to radius) similar to that of the construction giving our upper bound? It is equally natural to ask what the minimizer looks like; for example, must the amplitude and wavelength of wrinkling at radius $r$ resemble these of our construction?

It seems worth noting that while we have repeatedly used the Euler-Lagrange equation for the relaxed problem (to characterize the relaxed solution), we never used the Euler-Lagrange equation from the original problem. In related studies, such as~\cite{bib-bob+muller1}, minimizers are known to have special properties. We expect the same to be true in the present setting, but the analysis of minimizers will require new techniques. 

\section*{Appendix}

\renewcommand{\thelemma}{A.\arabic{lemma}}
\renewcommand{\theequation}{A.\arabic{equation}}

In this section we prove several lemmas which were used previously. Lemma~\ref{lm7} and Lemma~\ref{lm8} were used in the proof of the lower bound in the three-dimensional setting (Section~\ref{lb3D}). Poincar\' e inequality (Lemma~\ref{lmPoincare}) was used in the proof of the lower bound in the two-dimensional setting (Section~\ref{sect-lb2d}), Lemma~\ref{lm6} and Remark~\ref{rmk1} were used in the proof of Lemma~\ref{lm3}, and Lemma~\ref{lm5} was used in the definition of $\w$ (see~\eqref{eq049}).

\begin{lemma}\label{lm5}
Let $\wt$ be a stored energy function of an isotropic elastic material with
\begin{equation}\nonumber
 \wt(F) = g(I_1,I_2,J),
\end{equation}
where $J := \det (F)$, $C := F^TF$, and 
\begin{gather}\nonumber
 I_1 := J^{-2/3}\, \mathrm{tr}(C), \quad I_2 := \frac{J^{-4/3}}{2}\left( (\mathrm{tr}(C))^2 - \mathrm{tr}(C^2) \right).
\end{gather}
If $g$ satisfies 
\begin{gather}\label{eq-app}
 \frac{\partial g}{\partial I_1}(I_1,I_2,J) \ge 0, \quad \frac{\partial g}{\partial I_2}(I_1,I_2,J) \ge 0
\end{gather}
for $I_1 \ge 3$, $I_2 \ge 3$, and $J > 0$, then for any $M \in \R^{3\times 2}$ we have
\begin{equation}\nonumber
 \min_{\xi \in \R^3} \wt(M|\xi) = \min_{\xi \in \R^3, M \perp \xi} \wt(M|\xi),
\end{equation}
where $(M|\xi)$ denotes a $3 \times 3$ matrix with first two columns identical with $M$ and the third column $\xi$, and $M \perp \xi$ means the columns of $M$ are perpendicular to  $\xi$. 
If, moreover
\begin{gather}\label{eq-app2}
 \frac{\partial g}{\partial I_1}(I_1,I_2,J) + \frac{\partial g}{\partial I_2}(I_1,I_2,J) > 0
\end{gather}
for $I_1 \ge 3$, $I_2 \ge 3$, and $J > 0$, the minimum of $\wt(M|\xi)$ is attained only if $M \perp \xi$.
\end{lemma}

\begin{proof}
 Isotropy of the material implies rotational invariance of the energy density $\wt$, and so without loss of generality we can assume that
\begin{equation}\label{eq050}
 M = \begin{pmatrix} \lambda_1 & 0 \\ 0 & \lambda_2 \\ 0 & 0 \end{pmatrix}.
\end{equation}
Therefore we want to minimize
\begin{equation}\nonumber
 \wt \begin{pmatrix} \lambda_1 & 0 & x \\ 0 & \lambda_2 & y \\ 0 & 0 & z\end{pmatrix}
\end{equation}
among all possible $x,y,z$. A simple calculation reveals that
\begin{equation}\label{eq051}
\begin{aligned}
 J &= \lambda_1 \lambda_2 z, \\
 I_1 &= \left(\lambda_1 \lambda_2 z \right)^{-2/3} \left( \lambda_1^2 + \lambda_2^2 + x^2 + y^2 + z^2\right), \\
 I_2 &= \left(\lambda_1 \lambda_2 z \right)^{-4/3} \left( \lambda_1^2 \lambda_2^2 + \lambda_1^2 z^2 + \lambda_2^2 z^2 + \lambda_1^2 y^2 +  \lambda_2^2 x^2\right).
\end{aligned}
\end{equation}
Using the AM-GM inequality we see that $I_1 \ge 3$ and $I_2 \ge 3$ even if $x=y=0$. Fixing $z$ and varying $x$ and $y$ the value of $J$ stays constant whereas $I_1$ and $I_2$ are increasing functions of $x^2$ and $y^2$. Therefore \eqref{eq-app} implies that $\wt$ has its minimal value (for any $z$ fixed) if $x=y=0$. In the case~\eqref{eq-app2} the conclusion follows by the same reasoning.
\end{proof}

\begin{lemma}\label{lm7}
 Let $\wt$ satisfy~\eqref{eq031} and \eqref{eq049}. 
 Then for any $M \in \R^{3 \times 2}$ \footnotedva{Lemma~\eqref{lm7} - removed assumption $\det M > 0$} there exists a constant $C > 0$ such that
 \begin{equation}\nonumber
  \wt(M|v) - \w(M) = \wt(M|v) - \wt(M|\xi) \ge Cg_p\left(|v - \xi|\right) = Cg_p(| (M|v) - (M|\xi) | ),
 \end{equation}
 where $\xi = \argmin \wt(M|\ \cdot)$.
\end{lemma}

\begin{proof}
 If $|v - \xi|$ is large, the conclusion follows from the growth condition~\eqref{eq031}. Let us therefore assume that $|v - \xi| \le K$. 
 As in the proof of the previous lemma we may assume $M$ satisfies~\eqref{eq050}. We write
 \begin{gather}\nonumber
  v = (x,y,z)^T,\quad \xi = (0,0,Z)^T,
 \end{gather}
 where $\xi$ can be written in this form due to the previous lemma. We have
 \begin{equation}\nonumber
  \wt(M|v) - \wt(M|\xi) = \left(\wt(M|v) - \wt(M|w)\right) + \left(\wt(M|w) - \wt(M|\xi)\right),
 \end{equation}
 where 
 \begin{equation*}
 \quad w := (0,0,z)^T. 
 \end{equation*} 
 We estimate 
 \begin{equation}\label{eq052}
 \begin{aligned}
  \wt(M|v) - \wt(M|w) &= g(J,I_1,I_2) - g(J,\bar I_1, \bar I_2) \\ &= \nabla g \cdot (0,J^{-2/3}(x^2 + y^2), J^{-4/3}(\lambda_1^2 y^2 + \lambda_2^2 x^2)) \ge C(x^2 + y^2),
 \end{aligned}
 \end{equation}
 where $J,I_1,I_2$ are defined in~\eqref{eq051} and
\begin{equation}\nonumber
\begin{aligned}
 \bar I_1 &= \left(\lambda_1 \lambda_2 z \right)^{-2/3} \left( \lambda_1^2 + \lambda_2^2 + z^2\right), \\
 \bar I_2 &= \left(\lambda_1 \lambda_2 z \right)^{-4/3} \left( \lambda_1^2 \lambda_2^2 + \lambda_1^2 z^2 + \lambda_2^2 z^2 \right).
\end{aligned}
\end{equation}
We also have
\begin{equation}\label{eq053}
 \begin{aligned}
  \wt(M|w) - \wt(M|\xi) &= \ft(\lambda_1,\lambda_2,z) - \ft(\lambda_1,\lambda_2,Z) \\ &= \partial_{33}\ft (\lambda_1,\lambda_2,\zeta) (z-Z)^2 / 2  \ge C(z-Z)^2,
 \end{aligned}
 \end{equation}
\footnotedva{added $/2$ in A.10}
 where we have used that $\partial_3 \ft(\lambda_1,\lambda_2,Z) = 0$ and $D^2\ft > 0$ (see~\eqref{eq031}). Adding~\eqref{eq052} and~\eqref{eq053} we obtain the desired inequality.  
\end{proof}


\begin{lemma}\label{lm6}
 Let the density function $\wt$ satisfy \eqref{eq031} and \eqref{eq049}. For any unit vector $u_1 \in \R^3$ and another vector $v_1 \in \R^3$, $|v_1| > 1$, we define 
\begin{equation*}
 \widetilde F := \argmin \left\{ \wt(F) : F \in \R^{3\times 3} , F u_1 = v_1, \det F > 0 \right\}.
\end{equation*}
Then there exists an orthonormal basis $u_1,u_2,u_3$ and orthogonal vectors $v_1,v_2,v_3$ such that
\begin{equation}\label{eq007}
 \begin{gathered}
 \widetilde F u_i = v_i,\quad i=1,2,3,\\ |v_2| = |v_3| = g(|v_1|),
 \end{gathered}
\end{equation}
where $g(t)$ is a Lipschitz continuous function.
Moreover, there exist $0 < C_1 = C_1(v_1) < C_2=C_2(v_1)$ such that
 \begin{equation}\label{eq-quad3}
  C_1 g_p\left(\mathrm{dist} (G, SO(3)\widetilde F)\right) \le \wt (G) - \wt (\widetilde F) \le C_2\dist^2 (G, SO(3)\widetilde F),
\end{equation}
 for any $G$ satisfying $G u_1 = v_1$, $\det G > 0$ (function $g_p$ was defined in~\eqref{gp}).
\end{lemma}

\begin{remark}\label{rmk1}
As a consequence of~\eqref{eq-quad3} we obtain a similar condition for $\w$ -- for a given unit vector $n \in \mathbb{R}^2$ and another vector $m \in \mathbb{R}^3$ we have
\begin{equation}\nonumber
 C_1g_p(\mathrm{dist}( G, SO(3)F )) \le \w(G) - \w(F) \le C_2\dist^2( G, SO(3)F ),
\end{equation}
when $G n = F n = m$ and $\w(F) = \min\{ \w(H): H n=m \}$. 
\end{remark}

\begin{proof}[Proof of Lemma~\ref{lm6}]
 We first prove~\eqref{eq007}. Since $\wt$ is frame-indifferent and isotropic, we can assume WLOG that $u_1 = e_1$ and $v_1 = \lambda_1 e_1$, $\lambda_1 > 1$. Then we want to find all matrices $F$ that minimize
 \begin{equation}\nonumber
  \min \{ \wt(F) : F e_1 = \lambda_1 e_1 \}.
 \end{equation}
 Let $F$ be such a matrix. We see from Lemma~\ref{lm5} that the first column of $F$ and the third column of $F$ are perpendicular. Since $Fe_1 = \lambda_1e_1$, this means that $F_{13} = 0$. Since the second and third column are interchangeable, the first and second column are also perpendicular, i.e. $F_{12} = 0$. By the same lemma we also have that the second and third column of $F$ are perpendicular. Therefore, up to a rotation which fixes the first column, $F$ is diagonal:
 \begin{equation}\nonumber
  F = \begin{pmatrix}
       \lambda_1 & 0 & 0 \\
       0 & \lambda_2 & 0 \\
       0 & 0 & \lambda_3 
      \end{pmatrix}.
 \end{equation}
 We claim that $\lambda_2 = \lambda_3$. Indeed, strict convexity of $\ft$ implies
 \begin{multline}\nonumber
  \wt (F) = \ft(\lll) = \frac{\ft(\lll) + \ft(\lambda_1,\lambda_3,\lambda_2)}{2} \\ > \ft(\lambda_1, \frac{\lambda_2 + \lambda_3}{2}, \frac{\lambda_2 + \lambda_3}{2}) = \wt(F')
 \end{multline} 
 (where $F'e_1 = \lambda_1 e_1$), unless $\lambda_2 = \lambda_3$. Moreover, since
 \begin{equation}\nonumber
  \lambda_2 = \argmin_{t > 0} \ft(\lambda_1,t,t),
 \end{equation}
 it follows that $\lambda_2$ is a Lipschitz continuous function of $\lambda_1$. 
 
We observe that $\lambda_1 \ge \lambda_2$. Indeed we will show that 
\begin{equation}\label{l1l2}
\ft(\lambda_1,\lambda_1,\lambda_1) < \ft(\lambda_1,t,t) 
\end{equation}
for any $t > \lambda_1$. We compute \footnote{See~\eqref{eq049} for definition of $J,I_1,$ and $I_2$.}
\begin{align*}
 J(\lambda_1,t,t) = \lambda_1 t^2 &> \lambda_1^3 = J(\lambda_1,\lambda_1,\lambda_1),\\
 I_1(\lambda_1,t,t) = \lambda_1^{4/3} t^{-4/3} + 2\lambda_1^{-2/3} t^{2/3} &> 3 = I_1(\lambda_1,\lambda_1,\lambda_1),\\
 I_2(\lambda_1,t,t) = \lambda_1^{-4/3} t^{4/3} + 2\lambda_1^{2/3} t^{-2/3} &> 3 = I_2(\lambda_1,\lambda_1,\lambda_1),
\end{align*}
where the two latter inequalities follow from the AM-GM inequality. Therefore for $t > \lambda_1$ we have that $I_1(\lambda_1,t,t) > I_1(\lambda_1,\lambda_1,\lambda_1)$, $I_2(\lambda_1,t,t) > I_2(\lambda_1,\lambda_1,\lambda_1)$, $J(\lambda_1,t,t) > J(\lambda_1,\lambda_1,\lambda_1)$, and so~\eqref{l1l2} follows from~\eqref{eq049}.
 
To prove~\eqref{eq-quad3}, we fix the matrix $F$ (therefore $\lambda_1 \ge \lambda_2$ are also fixed). It is sufficient to show the lower bound and upper bound only for $G$ for which $\wt(G) - \wt(F)$ is small:
 \begin{equation}\label{eq039}
 C_1 \mathrm{dist}^2 (G, SO(3) F) \le \wt (G) - \wt (F) \le C_2 \mathrm{dist}^2 (G, SO(3) F)
 \end{equation}
(the upper bound and the lower bound for ``large'' $G$ follow from the growth assumptions on $\ft$ -- see~\eqref{eq031}).

We start with the lower bound. As before, we can assume $F$ is diagonal with entries $\lambda_1,\lambda_2,\lambda_3=\lambda_2$. Let $\sigma_1 \ge \sigma_2 \ge \sigma_3 \ge 0$ be singular values of $G$. We first see that
\begin{equation*}
 \lambda_1 = G_{11} = \mathrm{tr}(G e_1 \otimes e_1) \le \sigma_1 \cdot 1 = \sigma_1,
\end{equation*}
where the inequality follows from von Neumann's Lemma (Lemma~\ref{lm_neumann}), and $G_{11}$ denotes the upper left entry of matrix $G$. 
Since $\wt(G)$ depends only on the singular values of $G$, we will try to estimate both sides of~\eqref{eq-quad3} in terms of the singular values of $F$ and $G$. 

We start with estimating LHS of~\eqref{eq039}. We know
\begin{multline}\nonumber
 \mathrm{dist}^2(G,SO(3)F) = \min_{R \in SO(3)} \mathrm{tr}\left( (G - RF)^T(G- RF) \right) \\ 
 = \mathrm{tr}\left(G^TG\right) + \mathrm{tr}\left(F^TF\right) - 2\max_{R \in SO(3)} \mathrm{tr}\left(RFG^T\right).
\end{multline}
We claim that $\max_{R \in SO(3)} \mathrm{tr}\left(RFG^T\right)$ is equal to the sum of singular values of $FG^T$. To show this, we use the singular value decomposition $FG^T = UDV^T$, where $U,V \in SO(3)$ and $D$ is a diagonal matrix (with the singular values of $FG^T$ as diagonal entries). Then $\mathrm{tr}\left(RFG^T\right) = \mathrm{tr}\left(RUDV^T\right) = \mathrm{tr}\left(V^TRU D\right) = \mathrm{tr}\left(Q D\right)$. Since $Q$ is a rotation, we have that 
$\mathrm{tr}\left(Q D\right) \le \mathrm{tr}\left(D\right)$, and the maximum is attained (for $Q = I$). 

The sum of singular values of $FG^T$ is by definition equal to $$\mathrm{tr} \left( \left( FG^T GF^T \right)^{1/2} \right) = \mathrm{tr} \left( \left(G^TG\right)^{1/2}F \right).$$ We write $F = (\lambda_1 - \lambda_2) e_1 \otimes e_1 + \lambda_2 I$ and obtain
\[ \mathrm{tr} \left( \left(G^TG\right)^{1/2}F \right)\! =\! (\lambda_1 - \lambda_2) \left(\sqrt{G^TG}\right)_{11} + \lambda_2 \mathrm{tr} \left( \sqrt{G^TG} \right) = (\lambda_1 - \lambda_2) \left(\sqrt{G^TG}\right)_{11} + \lambda_2 \left( \sigma_1 + \sigma_2 + \sigma_3\right). \]
We know that $\sqrt{G^TG}$ has eigenvalues $\sigma_1,\sigma_2,\sigma_3$, and so $\mathrm{tr}(G^TG) = \sigma_1^2 + \sigma_2^2 + \sigma_3^2$. 
Therefore we see that
\begin{equation}\label{eq046}
\begin{aligned}
 \mathrm{dist}^2(G,SO(3)F) &= (\sigma_1^2 + \sigma_2^2 + \sigma_3^2) + (\lambda_1^2 + \lambda_2^2 + \lambda_2^2) - 2\lambda_2 (\sigma_1 + \sigma_2 + \sigma_3) - 2(\lambda_1 - \lambda_2)\alpha \\ &= (\sigma_1 - \lambda_1)^2 + (\sigma_2 - \lambda_2)^2 + (\sigma_3 - \lambda_2)^2 + 2(\lambda_1 - \lambda_2)(\sigma_1 - \alpha),
\end{aligned}
\end{equation}
where $\alpha$ denotes
\begin{equation}\label{eq040}
\alpha = (\sqrt{G^TG})_{11} = (\sqrt{G^TG}e_1,e_1). 
\end{equation}

We claim that 
\begin{equation}\label{eq045}
 \alpha \ge \frac{\lambda_1^2}{\sigma_1}.
\end{equation}
Indeed, we write $G = UDV$, where $U,V \in SO(3)$ and $D$ is a diagonal matrix with entries $\sigma_i, i=1,2,3$. We define a unit vector
$$x = (x_1,x_2,x_3) := Ve_1.$$
By virtue of~\eqref{eq040} we get
\begin{equation}\nonumber
 \alpha = (\sqrt{G^TG}e_1,e_1) = (V^TDVe_1,e_1) = (D Ve_1, Ve_1) = \sum_{i=1}^3 x_i^2 \sigma_i.
\end{equation}
We also know that $G e_1 = \lambda_1 e_1$, and so 
\begin{equation}\nonumber
 \lambda_1 e_1 = UDVe_1 = U\left(\sum_{i=1}^3 \sigma_i x_i e_i\right).
\end{equation}
In particular, since $U \in SO(3)$, we see that the norm of the vector on the RHS is $\lambda_1$, i.e.
\begin{equation}\label{eq048}
 \sum_{i=1}^3 x_i^2 \sigma_i^2 = \lambda_1^2. 
\end{equation}
To summarize, we have
\begin{gather}\label{eq043}
 \alpha = \sum_{i=1}^3 x_i^2 \sigma_i,\\ \sum_{i=1}^3 x_i^2 = 1 \textrm{ and } \sum_{i=1}^3 x_i^2 \sigma_i^2 = \lambda_1^2.\label{eq043a} 
\end{gather}
To find the lower bound for $\alpha$, we simply optimize $\sum_{i=1}^3 x_i^2 \sigma_i$ assuming \eqref{eq043a}. Using method of Lagrange multipliers we observe that for $\sum_{i=1}^3 x_i^2 \sigma_i$ to be minimal one $x_i$ has to be zero. We see that $\sum_{i=1}^3 x_i^2 \sigma_i$ is equal to:
\begin{equation}\nonumber
 \frac{\lambda_1^2 + \sigma_2\sigma_3 }{\sigma_2 + \sigma_3} \quad \textrm{ if } x_1 = 0, \qquad
 \frac{\lambda_1^2 + \sigma_1\sigma_3 }{\sigma_1 + \sigma_3} \quad \textrm{ if } x_2 = 0, \qquad
 \frac{\lambda_1^2 + \sigma_1\sigma_2 }{\sigma_1 + \sigma_2} \quad \textrm{ if } x_3 = 0. 
\end{equation}
Using convention $\sigma_1 \ge \sigma_2 \ge \sigma_3$ and \eqref{eq048} we see that $\sigma_1 \ge \lambda_1 \ge \sigma_3$. Since
\begin{equation*}
 \frac{\lambda_1^2 + \sigma_1\sigma_3 }{\sigma_1 + \sigma_3} \le \frac{\lambda_1^2 + \sigma_2\sigma_3 }{\sigma_2 + \sigma_3} \Longleftrightarrow \lambda_1^2 \ge \sigma_3^2
 \textrm{~~ and ~~}
 \frac{\lambda_1^2 + \sigma_1\sigma_3 }{\sigma_1 + \sigma_3} \le \frac{\lambda_1^2 + \sigma_1\sigma_2 }{\sigma_1 + \sigma_2} \Longleftrightarrow \lambda_1^2 \le \sigma_1^2, 
 \end{equation*}
the minimum of~\eqref{eq043} is equal to $\frac{\lambda_1^2 + \sigma_1\sigma_3 }{\sigma_1 + \sigma_3}$. Finally, we observe that $\sigma_1^2 \ge \lambda_1^2$ implies
\begin{equation*}
 \frac{\lambda_1^2 + \sigma_1\sigma_3 }{\sigma_1 + \sigma_3} \ge \frac{\lambda_1^2}{\sigma_1}.
\end{equation*}
We have proved~\eqref{eq045}.

Now we are ready to finish the proof of the lower bound. By virtue of~\eqref{eq046}:
\begin{equation}\nonumber
 \mathrm{dist}^2(G,SO(3)F) \le (\sigma_1 - \lambda_1)^2 + (\sigma_2 - \lambda_2)^2 + (\sigma_3 - \lambda_2)^2 + 2(\lambda_1 - \lambda_2)\left(\sigma_1 - \frac{\lambda_1^2}{\sigma_1}\right),
\end{equation}
For the middle term in~\eqref{eq039} we have:
\begin{align*}
 \wt(G) - \wt(F) &= \ft(\sigma_1,\sigma_2,\sigma_3) - \ft(\lambda_1,\lambda_2,\lambda_2) \\ &\ge \sum_{i=1}^3 \partial_i \ft (\lambda_1,\lambda_2,\lambda_2)(\sigma_i - \lambda_i) + C\left|(\sigma_1,\sigma_2,\sigma_3) - (\lambda_1,\lambda_2,\lambda_2)\right|^2 \\
 &= \partial_1 \ft (\lambda_1,\lambda_2,\lambda_2)(\sigma_1 - \lambda_1) + C\sum_{i=1}^3 (\sigma_i - \lambda_i)^2,
\end{align*}
where we have used strict convexity of $\ft$ (see~\eqref{eq031}) and $\partial_2 \ft(\lambda_1,\lambda_2,\lambda_2) = \partial_3 \ft (\lambda_1,\lambda_2,\lambda_2) = 0$ (a~consequence of the definition of $\lambda_2$). 
Therefore, it remains to show
\begin{equation}\label{eq047}
 2(\lambda_1 - \lambda_2)\left(\sigma_1 - \frac{\lambda_1^2}{\sigma_1}\right) \le C\, \partial_1 \ft (\lambda_1,\lambda_2,\lambda_2)\, (\sigma_1 - \lambda_1).
\end{equation}
We observe that since $\sigma_1 \ge \lambda_1$, we have $\left(\sigma_1 - \frac{\lambda_1^2}{\sigma_1}\right) \le 2(\sigma_1 - \lambda_1)$. Since $\lambda_1, \lambda_2$ are fixed and $\partial_1 \ft > 0$, \eqref{eq047} immediately follows (with constant $C$ depending on $\lambda_1$). 

We now turn to the proof of the upper bound. This is easy, since we have already done all the work. We know that
\begin{equation}\nonumber
\begin{aligned}
 \mathrm{dist}^2(G,SO(3)F) &= (\sigma_1 - \lambda_1)^2 + (\sigma_2 - \lambda_2)^2 + (\sigma_3 - \lambda_2)^2 + 2(\lambda_1 - \lambda_2)(\sigma_1 - \alpha) \\ &\ge (\sigma_1 - \lambda_1)^2 + (\sigma_2 - \lambda_2)^2 + (\sigma_3 - \lambda_2)^2 + 2(\lambda_1 - \lambda_2)(\sigma_1 - \lambda_1),
 \end{aligned}
\end{equation}
where we have used $\alpha \le \lambda_1$. This is true by~\eqref{eq043} and \eqref{eq043a}:
\begin{equation*}
 \alpha = \sum_{i=1}^3 x_i \left(x_i \sigma_i\right) \le \left(\sum_{i=1}^3 x_i^2\right)^{1/2} \left(\sum_{i=1}^3 x_i^2 \sigma_i^2\right)^{1/2} = \lambda_1.
\end{equation*}
We also know 
\begin{align*}
  \wt(G) - \wt(F) &= \ft(\sigma_1,\sigma_2,\sigma_3) - \ft(\lambda_1,\lambda_2,\lambda_2) \\ &\le \sum_{i=1}^3 \partial_i \ft (\lambda_1,\lambda_2,\lambda_2)(\sigma_i - \lambda_i) + C\left|(\sigma_1,\sigma_2,\sigma_3) - (\lambda_1,\lambda_2,\lambda_2)\right|^2 \\
 &= \partial_1 \ft (\lambda_1,\lambda_2,\lambda_2)(\sigma_1 - \lambda_1) + C\sum_{i=1}^3 (\sigma_i - \lambda_i)^2,
\end{align*}
where we have used boundedness of $D^2 \ft$ (see~\eqref{eq031}).
To finish the proof, it remains to observe that
\begin{equation}\nonumber
 \partial_1 \ft (\lambda_1,\lambda_2,\lambda_2)(\sigma_1 - \lambda_1) \le C(\lambda_1 - \lambda_2)(\sigma_1 - \lambda_1)
\end{equation}
holds trivially (with $\lambda_1 > \lambda_2$ being fixed). 
\end{proof}


\begin{lemma}\label{lm8}
 Let $F, G \in \R^{3\times 2}$. Let $\xi \in \R^3$ satisfies $\xi \perp F$, $\det (F|\xi) > 0$, and $|\xi| = l(F)$, and similarly let $\zeta \in \R^3$ satisfies $\zeta \perp G$, $\det (G|\zeta) > 0$, $|\zeta| = l(G)$, where $l(A)$ is a Lipschitz continuous function of singular values of $A$. Then there exists constant $C$ such that
 \begin{equation}\nonumber
  \dist( (F|\xi), SO(3)(G|\zeta) ) \le C\dist(F,SO(3)G).
 \end{equation} 
\end{lemma}

\begin{proof}
 Without loss of generality we assume columns of $F$ lie in the $x$-$y$ plane, i.e. $F_{31}=F_{32}=0$. Let $R \in SO(3)$ be such that
 \begin{equation}\label{eq060}
  \dist(F,SO(3)G) = |F - RG|.
 \end{equation}
 We show that $R$ can be chosen such that columns of $RG$ lie in the $x$-$y$ plane as well. We know that (see the proof of Lemma~\ref{lm6}):
 \begin{equation}\nonumber
  \dist^2(F,SO(3)G) = \mathrm{tr} (F^TF) + \mathrm{tr} (G^TG) - 2\mathrm{tr} \left( (FG^TGF^T)^{1/2}\right).
 \end{equation}
 Since $F_{31}=F_{32}=0$, the RHS in the last relation does not change if we replace $F$ by its first two rows. The last term is then equal to $\mathrm{tr} \left( (G^TG)^{1/2} (F^TF)^{1/2} \right)$, and we obtain
 \begin{equation}\nonumber
  \dist^2(F,SO(3)G) = || (G^TG)^{1/2} - (F^TF)^{1/2} ||^2. 
 \end{equation}
 
 Now let $S \in SO(3)$ be such that $(SG)_{31} = (SG)_{32} = 0$. We want to compute 
 \begin{equation}\nonumber
  \dist(F,SO(2)SG),
 \end{equation}
 where we treat $F$ and $SG$ as $2 \times 2$ matrices (since both their third rows vanish). 
 Following the previous reasoning we obtain
 \begin{equation}\nonumber
  \dist^2(F,SO(2)SG) = || ((SG)^TSG)^{1/2} - (F^TF)^{1/2} ||^2 = || (G^TG)^{1/2} - (F^TF)^{1/2} ||^2.
 \end{equation}
 We have shown that $\dist^2(F,SO(2)SG) = \dist^2(F,SO(3)G)$, i.e. that $R$ can be chosen such that $RG$ lies in the $xy$ plane. 
 
 We have
\begin{multline}\nonumber 
  \dist^2( (F|\xi), SO(3)(G|\zeta) ) \le || (F|\xi) -  R (G|\zeta) ||^2 \\ = || F - RG||^2 + |\xi - R\zeta|^2 = \dist^2( F, SO(3)G ) + |\xi - R\zeta|^2.
 \end{multline} 
 Since $F$ and $RG$ lie in the $x$-$y$ plane, both $\xi$ and $R\zeta$ are perpendicular to this plane.  It is straightforward but tedious to show that in fact $\xi$ and $R\zeta$ have the same orientation. Then we just use the fact that $|\xi| = l(F)$ and $|\zeta| = l(G)$ together with Lipschitz continuity of $l$ to obtain
 \begin{equation}\nonumber 
  |\xi - R\zeta| = |l(F) - l(RG)| \le C|F - RG| = C\dist(F,SO(3)G).
 \end{equation}
 We are done since~\eqref{eq060} and the previous relation imply
 \begin{equation}\nonumber
   \dist^2( (F|\xi), SO(3)(G|\zeta) ) \le |(F|\xi) - R(G|\zeta)|^2 = |F - RG|^2 + |\xi - R\zeta|^2 \le C\dist^2(F,SO(3)).
 \end{equation} 
\end{proof}


\begin{lemma}[Poincar\'e inequality]\label{lmPoincare}
Let $g_p$ be as in~\eqref{gp} with $1 < p \le 2$. Then there exists a constant $C(U,p)$ such that for every $v \in W^{1,p}(U)$ there exists a constant $\bar v$ and
\begin{equation}\nonumber\label{a1}
 \int_U g_p(|v - \bar v|) \dxyz \le C \int_U g_p(|\nabla v|) \dxyz.
\end{equation}
\end{lemma}

\begin{proof}
 We first observe that since $g_p(t) \le \frac{1}{2} \min (t^p,t^2)$ and $g_p$ is convex, there exists $C$ such that 
 \begin{equation}\label{a4}
  g_p(s+t) \le C(s^p + t^2),\quad \textrm{ for every } s,t \ge 0.
 \end{equation}
 In the proof we will use the following truncation result proved in~\cite{bib-rigid1}:
 
 \begin{prop*}[Truncation]
  Suppose $U \subset \R^n$ is a bounded Lipschitz domain. Then there exists a constant $C(U,p)$ with the following property:
  For each $v \in W^{1,p}(U)$ and each $\lambda > 0$, there exists $V : U \to \R$ such that
  \begin{equation}\nonumber
  \begin{array}{cl}
   \mathrm{(i)} & \quad ||\nabla V||_{L^\infty} \le C\lambda \\
   \mathrm{(ii)} & \quad \left|\{ x \in U : v(x) \neq V(x)\}\right| \le \frac{C}{\lambda^p} \int_{\{ x\in U:|\nabla v(x)|>\lambda\}} |\nabla v|^p \dxyz,\\
   \mathrm{(iii)} & \quad ||\nabla v - \nabla V||_{L^p(U)}^p \le C \int_{ \{ x \in U:|\nabla v(x)|>\lambda\} } |\nabla v|^p \dxyz.
  \end{array}
  \end{equation}
 \end{prop*}
 
Let us denote 
\begin{equation}\nonumber
  K := \int_U g_p(|\nabla v|) \dxyz.
 \end{equation}
By the proposition with $\lambda = 1$ there exists $V \in W^{1,\infty}$ such that $|\nabla V| \le C$ and 
\[||\nabla v - \nabla V||_{L^p}^p \le C\int_{\{|\nabla v|>1\}} |\nabla v|^p \dxyz \le CK.\]
%
The standard Poincar\'e inequality implies
 \begin{multline}\label{a2}
  \int_U |V - \bar V|^2 \dxyz \le C \int_U |\nabla V|^2 \dxyz \le C\left( \int_{\{ v \neq V \}} |\nabla V|^2 \dxyz + \int_{\{ v = V \}} |\nabla V|^2 \dxyz\right) \\ \le C\left|  \left\{ v \neq V \right\} \right| + C\int_{\{ v = V \}} |\nabla V|^2 \dxyz \le CK.
 \end{multline}
We also get
\begin{equation}\label{a3}
 \int_U |V - v - a|^p \dxyz \le C\int_U |\nabla V - \nabla v|^p \dxyz \le CK.
\end{equation}
Using~\eqref{a4},~\eqref{a2}, and~\eqref{a3} we obtain
\begin{equation}\nonumber
 \int_U g_p\left(|v - (a + \bar V)|\right) \dxyz \le C\left( \int_U |V-\bar V|^2 \dxyz + \int_U |V - v - a|^p \dxyz\right) \le CK. 
\end{equation}
\end{proof}

\bibliographystyle{amsplain}
\bibliography{bella}

\end{document}